\documentclass[11pt, fleqn]{article}
\usepackage[english]{babel}

\usepackage[lmargin=1.1in,rmargin=1.1in,bottom=1.3in,top=1.3in,
twoside=False]{geometry}

\usepackage{relsize,xspace}
 \usepackage{xcolor}
 \usepackage{mathtools}
 \usepackage{todonotes}
 \usepackage{comment}
\usepackage{microtype}
\usepackage{amsmath}
\usepackage{amssymb}
\usepackage{amsfonts}
\usepackage{stmaryrd}
\usepackage{bm}
\usepackage{tikz}
\usepackage{refcount}
\usepackage{wrapfig}
\usepackage{thmtools}
\usepackage{thm-restate}

\usepackage{marginnote}

\definecolor{sz}{rgb}{0.1,0.2,0.6}
\definecolor{blue}{rgb}{0.1,0.2,0.5}
\definecolor{brown}{rgb}{0.6,0.6,0.2}

\usepackage[amsmath,thmmarks,hyperref]{ntheorem}
\usepackage{cleveref}

\crefformat{page}{#2page~#1#3}%
\Crefformat{page}{#2Page~#1#3}%
\crefformat{equation}{#2(#1)#3}%
\Crefformat{equation}{#2(#1)#3}%
\crefformat{figure}{#2Figure~#1#3}%
\Crefformat{figure}{#2Figure~#1#3}%
\crefformat{section}{#2Section~#1#3}
\Crefformat{section}{#2Section~#1#3}
\crefformat{chapter}{#2Chapter~#1#3}
\Crefformat{chapter}{#2Chapter~#1#3}
\crefformat{chapter*}{#2Chapter~#1#3}
\Crefformat{chapter*}{#2Chapter~#1#3}
\crefformat{part}{#2Part~#1#3}
\Crefformat{part}{#2Part~#1#3}
\crefformat{enumi}{#2(#1)#3}
\Crefformat{enumi}{#2(#1)#3}

\usepackage{enumerate}

\usepackage{latexsym}

\theoremnumbering{arabic}
\theoremstyle{plain}
\theoremsymbol{}
\theorembodyfont{\itshape}
\theoremheaderfont{\normalfont\bfseries}
\theoremseparator{.}

\newtheorem{theorem}{Theorem}
\crefformat{theorem}{#2Theorem~#1#3}
\Crefformat{theorem}{#2Theorem~#1#3}
\renewcommand{\setminus}{-}
\newcommand{\newtheoremwithcrefformat}[2]{%
  \newtheorem{#1}[lemma]{#2}%
  \crefformat{#1}{##2\MakeUppercase#1~##1##3}%
  \Crefformat{#1}{##2\MakeUppercase#1~##1##3}%
}
\newcommand{\newseptheoremwithcrefformat}[2]{%
  \newtheorem{#1}{#2}%
  \crefformat{#1}{##2\MakeUppercase#1~##1##3}%
  \Crefformat{#1}{##2\MakeUppercase#1~##1##3}%
}

\newseptheoremwithcrefformat{lemma}{Lemma}
\newtheoremwithcrefformat{proposition}{Proposition}
\newtheoremwithcrefformat{observation}{Observation}
\newtheoremwithcrefformat{conjecture}{Conjecture}
\newtheoremwithcrefformat{corollary}{Corollary}
\newseptheoremwithcrefformat{claim}{Claim}
\newseptheoremwithcrefformat{goal}{Goal}
\theorembodyfont{\upshape}
\newtheoremwithcrefformat{example}{Example}
\newtheoremwithcrefformat{remark}{Remark}
\newseptheoremwithcrefformat{definition}{Definition}

\theoremstyle{nonumberplain}
\theoremheaderfont{\scshape}
\theorembodyfont{\normalfont}
\theoremsymbol{\ensuremath{\square}}
\newtheorem{proof}{Proof}

\theoremsymbol{\ensuremath{\lrcorner}}
\newtheorem{clproof}{Proof}

\newcommand{\set}[1]{\{#1\}}
\newcommand{\setof}[2]{\left\{#1 \,\mid\, #2 \right\}}
\renewcommand{\subset}{\subseteq}

\newcommand{\Oof}{\mathcal{O}}
\newcommand{\CCC}{\mathcal{C}}

\newcommand{\FFF}{\mathcal{F}}
\newcommand{\GGG}{\mathcal{G}}

\newcommand{\Pow}{\mathcal{P}}
\newcommand{\N}{\mathbb{N}}

\newcommand{\tup}[1]{\bar{#1}}
\renewcommand{\phi}{\varphi}
\renewcommand{\epsilon}{\varepsilon}
\newcommand{\str}{\mathbb}
\newcommand{\strA}{\str{A}}

\newcommand{\minor}{\preccurlyeq}
\newcommand{\dist}{\mathrm{dist}}

\renewcommand{\mid}{~:~}

\newcommand{\abs}[1]{\ensuremath{\left\lvert#1\right\rvert}}

\newcommand{\from}{\colon}

\renewcommand{\leq}{\leqslant}
\renewcommand{\geq}{\geqslant}
\renewcommand{\le}{\leqslant}
\renewcommand{\ge}{\geqslant}

\hypersetup{ocgcolorlinks=true,colorlinks=true,linkcolor={blue}, citecolor={brown}}

\newcounter{aux}

\title{On the number of types in sparse graphs
\thanks{
The work of M.\ Pilipczuk and S.\ Siebertz is supported by the National Science Centre of 
Poland via POLONEZ grant agreement UMO-2015/19/P/ST6/03998, 
which has received funding from the European Union's Horizon 2020 research and 
innovation programme (Marie Sk\l odowska-Curie grant agreement No.\ 665778).
The work of Sz.~Toru{\'n}czyk is supported by the National Science Centre of Poland grant 2016/21/D/ST6/01485.
M. Pilipczuk is supported by the Foundation for Polish Science (FNP) via the START stipend programme.
}}

\author{
Micha\l~Pilipczuk \qquad
\qquad Sebastian Siebertz
\qquad Szymon Toru{\'n}czyk\\[0.3cm]
Institute of Informatics, University of Warsaw, Poland\\[0.1cm]
\texttt{\{michal.pilipczuk,siebertz,szymtor\}@mimuw.edu.pl}}

\begin{document}

\maketitle
\begin{abstract}
\noindent 
% We initiate the study of {\em{VC-density}} in nowhere dense classes of graphs; nowhere denseness
% is an abstract notion of sparsity introduced by Ne\v set\v ril and Ossona de Mendez~\cite{nevsetvril2010first,nevsetvril2011nowhere}.
% Our main result states
We prove that for every class of graphs $\CCC$ which is nowhere dense, as defined by  Ne\v set\v ril and Ossona de Mendez~\cite{nevsetvril2010first,nevsetvril2011nowhere},
and for every  first order formula $\phi(\tup x,\tup y)$, 
whenever one draws a graph $G\in \CCC$ and a subset of its nodes $A$, the number of subsets of $A^{|\tup y|}$ which are of the form $\set{\tup v\in A^{|\tup y|}\, \colon\, G\models\phi(\bar u,\tup v)}$ for some valuation $\tup u$ of $\tup x$ in $G$
is bounded by $\Oof(|A|^{|\tup x|+\epsilon})$, for every $\epsilon>0$. 
This provides optimal bounds on the VC-density of first-order definable set systems in nowhere dense graph classes.

We also give two new proofs of upper bounds on quantities in nowhere dense classes which are relevant for their logical treatment.
Firstly, we provide a new proof of the fact that nowhere dense classes are uniformly
quasi-wide, implying explicit, polynomial upper bounds on the functions
relating the two notions.  Secondly, we give a new combinatorial proof
of the result of Adler and Adler~\cite{adler2014interpreting} stating
that every nowhere dense class of graphs is stable. In contrast to the previous proofs of the above results,
our proofs are completely finitistic and constructive, and yield explicit and computable upper
bounds on quantities related to uniform quasi-wideness (margins) and stability (ladder indices).
\end{abstract}

\begin{picture}(0,0) \put(397,-217)
{\hbox{\includegraphics[scale=0.25]{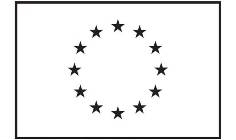}}} \end{picture} 
\vspace{-0.8cm}

\section{Introduction}\label{sec:intro}

Nowhere dense classes of graphs were introduced 
by Ne\v set\v ril and Ossona de 
Mendez~\cite{nevsetvril2010first,nevsetvril2011nowhere} as a 
general and abstract model
capturing uniform sparseness of graphs. These classes generalize many 
familiar classes of sparse graphs, such as planar graphs, graphs 
of bounded treewidth,  graphs of bounded degree, and, in fact, 
all classes that exclude a fixed 
topological minor.
Formally, a class $\CCC$ of graphs is {\em{nowhere dense}} if there is a function $t\colon \N\to \N$ such that for every $r\in \N$, 
no graph $G$ in~$\CCC$ contains the clique $K_{t(r)}$ on $t(r)$ vertices  as  {\em{depth-$r$ minor}},
i.e., as a subgraph of a graph obtained from $G$ by contracting mutually disjoint  subgraphs of radius at most $r$ to single vertices.
The more restricted notion of {\em{bounded expansion}} requires in addition that for every fixed $r$, there is a constant (depending on $r$) upper bound on the ratio 
between the number of edges and the number of vertices in depth-$r$ minors of graphs from $\CCC$.

The concept of nowhere denseness
turns out to be very robust, as witnessed by the fact that 
it admits multiple different characterizations, uncovering intricate connections to seemingly distant branches of mathematics.
For instance,  nowhere dense graph classes can be characterized 
by upper bounds on the density of bounded-depth (topological) 
minors~\cite{nevsetvril2010first,nevsetvril2011nowhere},
by uniform quasi-wideness~\cite{nevsetvril2011nowhere} (a notion introduced by
Dawar~\cite{dawar2010homomorphism} in the context of homomorphism
preservation properties), by low tree-depth
colorings~\cite{nevsetvril2008grad}, by generalized coloring
numbers~\cite{zhu2009coloring}, by sparse neighborhood
covers~\cite{GroheKRSS15,grohe2014deciding}, by a game called the
splitter game~\cite{grohe2014deciding}, and by the model-theoretic
concepts of stability and independence~\cite{adler2014interpreting}.
For a broader discussion on graph theoretic sparsity we refer to the book
of Ne\v{s}et\v{r}il and Ossona de Mendez~\cite{sparsity}.

The combination of combinatorial and logical methods yields a powerful toolbox for the study
of nowhere dense graph classes. In particular, 
the result of Grohe, Kreutzer and the second author~\cite{grohe2014deciding} exploits 
this combination in order to prove that  a given
first order sentence $\varphi$ can be evaluated in time 
$f(\varphi)\cdot n^{1+\epsilon}$ on $n$-vertex graphs from a fixed nowhere dense class of graphs $\CCC$, for any fixed real $\epsilon>0$ and some function $f$.
On the other hand, provided $\CCC$ is closed under taking subgraphs, it is known that if $\CCC$ is not nowhere dense,
then there is no algorithm with running time of the form $f(\varphi)\cdot n^c$ for any constant~$c$ under plausible complexity assumptions~\cite{dvovrak2013testing}.
In the terminology of parameterized complexity, these results show that the notion of nowhere denseness exactly characterized subgraph-closed classes where model-checking first order logic
is fixed-parameter tractable, and conclude a long line of research concerning the parameterized complexity of the model checking problem for sparse graph classes (see \cite{grokre11} for a survey).

\paragraph{Summary of contribution.} In this paper, we continue the study of the 
interplay of combinatorial and logical properties
of nowhere dense graph classes, and provide
new upper bounds on several
quantities appearing in their logical study.
Our main focus is on the notion of \emph{VC-density} for first order formulas. This concept originates from model theory and 
aims to measure the complexity of set systems definable by first order formulas, similarly to the better-known VC-dimension.
We give optimal bounds on the VC-density in nowhere dense graph classes, and in particular we show that these bounds are ``as good as one could hope for''.

We also provide new upper bounds on quantities related to {\em{stability}} and {\em{uniform quasi-wideness}} of nowhere dense classes.
For stability, we provide explicit and computable upper bounds on the \emph{ladder index} of any first order formula on a given nowhere dense class.
For uniform quasi-wideness, we give a new, purely combinatorial proof of polynomial upper bounds on {\em{margins}}, that is, functions governing this notion.
We remark that the existence of upper bounds as above is known~\cite{adler2014interpreting,siebertz2016polynomial}, but the proofs are based on nonconstructive arguments, 
notably the compactness theorem for first order logic. Therefore, the upper bounds are given purely existentially and are not effectively computable.
Contrary to these, our proofs are entirely combinatorial and effective, yielding computable upper bounds.

We now discuss the relevant background from logic and model theory, in order to motivate and state our results.

\paragraph{Model theory.}Our work is inspired by ideas from model theory,  more specifically, from \emph{stability theory}.
  The goal of {stability theory}
  is to draw certain dividing lines
  specifying abstract properties of 
  logical structures which allow the development 
  of a good structure theory. There are many such dividing lines, depending on the specifics of the desired theory. One such dividing line encloses the class of \emph{stable structures}, 
  another encloses the larger class of \emph{dependent structures} (also called \emph{NIP}). 
  A general theme is that the existence of a manageable structure is strongly related to
  the non-existence of certain forbidden patterns on one hand,
and on the other hand, to bounds on cardinalities
of certain \emph{type sets}.  
  Let us illustrate this phenomenon more concretely.

For a first order formula 
$\phi(\tup{x},\tup{y})$ 
 with free variables
split into  $\tup{x}$ and $\tup{y}$,
a \emph{$\phi$-ladder}
of length $n$ in a logical structure $\str A$ is a sequence $\tup{u}_1,\ldots, \tup{u}_{n},
\tup{v}_1,\ldots, \tup{v}_{n}$ of tuples of elements of $\str A$ 
such that 
$$\strA\models\phi(\tup{u}_i,\tup{v}_j)\ \Longleftrightarrow\ i\leq j\qquad \text {for all $1\leq i,j\le n$.}$$
The least  $n$ for which 
there is no $\phi$-ladder of length $n$ is 
the \emph{ladder index} 
of $\phi(\tup{x},\tup{y})$ in~$\str A$ (which may depend on the split of the
variables, and may be $\infty$ for some infinite structures $\str A$). A class of structures $\CCC$ is \emph{stable} if
the ladder index of every first order formula $\phi(\tup{x},\tup{y})$ over
structures from~$\CCC$ is bounded by a constant depending only on $\phi$ 
and~$\CCC$. This notion can be applied to a single infinite structure $\str A$, by considering the class consisting of $\str A$ only.
Examples of stable structures include $(\str N,=)$,
the field of complex numbers $(\str C,+,\times,0,1)$,
as well as any vector space $V$ over the field of rationals, treated as a group with addition. On the other hand, $(\str Q,\le)$ and the field of reals $(\str R,+,\times,0,1)$ are not stable,
as they admit a linear ordering which is definable by a first order formula.
Stable structures turn out to have more graspable  structure than unstable ones, as they can be equipped with various notions 
useful for their study, such as
\emph{forking independence} (generalizing linear independence in vector spaces)
and \emph{rank} (generalizing dimension).
We refer to the textbooks~\cite{pillay,tent2012course} for an introduction to stability theory.

  One of concepts studied in the early 
  years of stability theory is 
  a property of infinite graphs  called \emph{superflatness}, introduced by Podewski and Ziegler~\cite{podewski1978stable}.
  The definition of superflatness is the same as   of nowhere denseness, but 
   Podewski and Ziegler,
  instead of applying it to an infinite class of finite graphs, apply it to a single infinite graph.
  The main result of~\cite{podewski1978stable} is that every superflat graph is stable.   
As observed by Adler and Adler~\cite{adler2014interpreting}, 
this directly implies  the following:
 \begin{theorem}[\cite{adler2014interpreting,podewski1978stable}]\label{thm:adleradler}
 	Every nowhere dense class of graphs is stable. Conversely, any stable class of finite graphs which is subgraph-closed  is nowhere dense.
 \end{theorem}
 Thus, the notion of nowhere denseness (or superflatness) coincides with stability if we restrict attention to subgraph-closed graph classes.
 
The proof of Adler and Adler does not yield effective or computable upper bound on the 
ladder index of a given formula for a given nowhere dense class of graphs, as it relies on the result of Podewski and Ziegler, which in turn invokes compactness for first order logic.

\paragraph{Cardinality bounds.}
One of the key insights provided by the work of Shelah is that stable classes can be characterized by admitting strong upper bounds on the cardinality of \emph{Stone spaces}.
For a first order formula $\phi(\tup x,\tup y)$ 
with free variables partitioned into \emph{object variables} $\bar x$ and \emph{parameter variables} $\tup y$, a logical structure $\str A$,
and a subset of its domain $B$, define
the set of \emph{$\phi$-types} with parameters from~$B$, which are realized in 
$\strA$, as follows\footnote{Here, $S^\phi(\str A/B)$ is the set  of types which are \emph{realized} in $\str A$. In model theory,
one usually works with the larger class of \emph{complete types}. This distinction will not be relevant here.}:
\begin{equation}\label{eq:stone-def}
S^\phi(\strA/B)=\left\{\big\{\tup v\ \in B^{|\bar y|}\, \colon\, \strA\models\phi(\tup u,\tup v)\big\} \colon\, \tup u\in V(\strA)^{|\bar x|}\right\}\ \subset\  \Pow(B^{|\bar y|}).
\end{equation}
Here, $V(\strA)$ denotes the domain of $\strA$ and $\Pow(X)$ denotes the powerset of $X$.
Putting the above definition in words, every tuple $\tup u\in V(\strA)^{|\bar x|}$ defines the set of those tuples $\tup v\in B^{|\bar y|}$ for which $\phi(\tup u,\tup v)$ holds.
The set $S^\phi(\strA/B)$ consists of all subsets of $B^{|\bar y|}$ that can be defined in this way.

Note that in principle, $S^\phi(\str A/B)$
may be equal to $\Pow(B^{|\bar x|})$, and therefore have very large cardinality compared to $B$, even for very simple formulas. 
The following characterization due to Shelah 
(cf. \cite[Theorem 2.2, Chapter II]{shelah1990classification})
shows that for stable classes this does not happen.
\begin{theorem}\label{thm:Shelah-stone-space}
A class of structures $\cal C$
	is stable if and only if 
	there is 
	an infinite cardinal $\kappa$
	such that the following holds for all
	structures
	$\str A$ in the elementary closure\footnote{The elementary closure of $\cal C$ is 
	the class of all structures $\str A$
	such that  every first order sentence $\phi$
	which holds in $\str A$ also holds in some $\str B\in \cal C$. Equivalently, it is the class of 
   models of the theory of $\cal C$.} of~$\cal C$, and all~$B\subset V(\str A)$:
$$  \textit{if\quad}|B|\le \kappa\textit {\quad then\quad}
|S^\phi(\str A/B)|\le \kappa.$$
\end{theorem}
Therefore,
Shelah's result provides an upper bound on the number of types, albeit using infinite cardinals, elementary limits, and infinite parameter sets.
 The cardinality bound provided by \cref{thm:Shelah-stone-space},
 however, does not seem to  immediately translate to a result of finitary nature. As we describe below,
 this can be done using the notions of {\em{VC-dimension}} and {\em{VC-density}}.

\paragraph{VC-dimension and VC-density.} The notion of VC-dimension was introduced by 
Vapnik and Chervonenkis~\cite{chervonenkis1971theory} as a measure of complexity of set systems, or equivalently of hypergraphs.
Over the years it
has found important applications in many areas of
statistics, discrete and computational geometry, 
and learning theory. 

Formally, VC-dimension is defined as follows. 
Let $X$ be a set and let  $\FFF\subseteq \Pow(X)$ 
be a family of subsets of $X$.
A subset $A\subseteq X$ is \emph{shattered by $\FFF$} if
$\{A\cap F\colon F\in \FFF\}=\Pow(A)$; that is, every subset of $A$ can be obtained as the intersection of some set from $\FFF$ with $A$. 
The \emph{VC-dimension},
of $\FFF$ is the maximum size of a subset $A\subseteq X$ that is shattered by
$\FFF$.

As observed by Laskowski~\cite{laskowski1992vapnik}, VC-dimension can be connected to concepts from stability theory introduced by Shelah.
For a given structure $\str A$, parameter set $B\subseteq V(\str A)$, and formula $\phi(\bar x,\bar y)$,
we may consider the family $S^\phi(\str A/B)$ of subsets of $B^{|\tup y|}$ defined using equation~\eqref{eq:stone-def}.
The \emph{VC-dimension} of $\phi(\bar x,\bar y)$ on $\str A$ is the VC-dimension of the family $S^\phi(\str A/V(\str A))$. 
In other words, the VC-dimension of $\phi(\bar x,\bar y)$
on $\str A$ is the largest cardinality of a finite
set $I$ for which there exist families of tuples $(\bar u_i)_{i\in I}$ and $(\bar v_J)_{J\subset I}$
of elements of $\str A$
such that  $$\strA\models\phi(\tup{u}_i,\tup{v}_J)\Longleftrightarrow i\in J\qquad \text {for all $i\in I$ and $J\subset I$.}$$
A formula $\phi(\bar x,\bar y)$ is \emph{dependent} on a class of structures $\cal C$
if there is a bound $d\in\N$ such that the VC-dimension of $\phi(\bar x,\bar y)$ on $\str A$ is at most $d$ for all $\str A\in\cal C$.
It is immediate from the definitions  that if a formula $\phi(\bar x,\bar y)$ is stable over $\cal C$, then it is also dependent on $\cal C$ (the bound being the ladder index). 
A class of structures  $\cal C$ is {\em{dependent}} if every formula $\phi(\bar x,\bar y)$ is dependent over $\cal C$. 
In particular, every stable class is dependent, and hence, by \cref{thm:adleradler}, every nowhere dense class of graphs is dependent.
Examples of infinite dependent structures (treated as singleton classes) include 
$(\mathbb Q,\le )$ and the field of reals $(\mathbb R,\times,+,0,1)$. 

One of the main properties of VC-dimension is that it implies polynomial upper bounds on the number of different ``traces'' that a set system can have on a given parameter set.
This is made precise by the well-known Sauer-Shelah Lemma, stated as follows.
\begin{theorem}[Sauer-Shelah Lemma, \cite{chervonenkis1971theory,sauer1972density, shelah1972combinatorial}]\label{thm:sauer-shelah}
  For any family $\FFF$ of subsets of a set $X$, if the VC-dimension of $\FFF$ is $d$,
  then for every finite $A\subset X$,
$$|\setof{A\cap F}{F\in {\cal F}}|\le c\cdot |A|^d,  
\qquad\textit{where $c$ is a universal constant.}$$
\end{theorem}
In particular, this implies that 
in a dependent class of structures $\cal C$, 
for every formula $\phi(\bar x,\bar y)$
there exists some constant $d\in \N$
such that
\begin{equation}\label{eq:nip}
|S^\phi(\str A/B)|\le c\cdot |B|^d,	
\end{equation}
for all $\str A\in\cal C$ and finite $B\subset V(\str A)$.
Unlike~\cref{thm:Shelah-stone-space}, this result 
is of finitary nature: it provides quantitative upper bounds on the number of different definable subsets of a given finite parameter set. 
Together with~\cref{thm:adleradler}, this implies that for every nowhere dense class of graphs $\cal C$
and every first order formula $\phi(\bar x,\bar y)$,
there exists a constant $d\in\N$ such that~\eqref{eq:nip} holds. 

However, the VC-dimension $d$ may be enormous and it highly depends on $\cal C$ and the formula $\phi(\bar x,\bar y)$.
This suggests investigating quantitative bounds of the form \eqref{eq:nip} for exponents smaller than the VC-dimension $d$, as it is conceivable that the combination of bounding VC-dimension and applying
Sauer-Shelah Lemma yields a suboptimal upper bound. Our main goal is to decrease this exponent drastically in the setting of nowhere dense graph classes.

The above discussion motivates the notion of \emph{VC-density}, a notion closely related to VC-dimension.
The \emph{VC-density} (also called the 
\emph{VC-exponent})
of a set system $\cal F$
on an infinite set $X$ is the infimum of all reals $\alpha>0$ such that 
$|\setof{A\cap F}{F\in \cal F}|\in \Oof(|A|^\alpha)$, for all finite $A\subset X$. 
% In other words, this is the best exponent one could use in an application of the Sauer-
Similarly, the VC-density of a formula $\phi(\bar x,\bar y)$ over a class of structures~$\cal C$
is the infimum of all reals $\alpha>0$
such that $|S^\phi(\str A/B)|\in \Oof(|B|^\alpha)$,
for all $\str A\in \cal C$ and all finite $B\subset V(\str A)$.
 % of VC-densities of set families $S^\phi(\str A/B)$, for $\str A\in \cal C$ and finite parameter sets $B\subseteq V(\str A)$.
The Sauer-Shelah Lemma
implies that the VC-density (of a set system, or of a formula over a class of structures) is bounded from above by the VC-dimension. 
However, in many cases, the VC-density may be much smaller than the VC-dimension. Furthermore, it is the VC-density, rather than VC-dimension, that is actually relevant in combinatorial
and algorithmic applications~\cite{Bronnimann1995,matouvsek1998geometric,Matousek:2004:BVI:1005787.1005789}, see also \cref{sec:ep}.
We refer to~\cite{aschenbrenner2016vapnik} for an overview of 
applications of VC-dimension and VC-density in model
theory and to the surveys~\cite{furedi1991traces,matouvsek1998geometric} 
on uses of VC-density in
combinatorics. 

\paragraph{The main result.}
Our main result, \cref{thm:vc-density} stated below, improves the bound~\eqref{eq:nip} for classes of sparse graphs
by providing essentially the optimum exponent.

 \newcounter{vcupper}
 \setcounter{vcupper}{\thetheorem}
 \begin{theorem}\label{thm:vc-density}
%\begin{restatable}{theorem}{vcupper}\label{thm:vc-density}
Let $\CCC$ be a class of graphs and let $\phi(\tup x,\tup y)$ be a first order formula
with free variables  partitioned  into object variables $\bar x$  and parameter variables $\bar y$. Let $\ell=|\bar x|$. Then:
\begin{enumerate}[(1)]
\item If $\CCC$ is nowhere dense, then for every $\epsilon>0$ 
there exists a constant~$c$ such that for every $G\in \CCC$ and every nonempty
$A\subseteq V(G)$, we have $|S^\phi(G/A)|\leq c\cdot |A|^{\ell+\epsilon}.$
\item If $\CCC$ has bounded expansion, then there exists a constant~$c$ such that for every $G\in \CCC$ and every nonempty $A\subseteq V(G)$, we have $|S^\phi(G/A)|\leq c\cdot |A|^\ell$.
\end{enumerate}
%\end{restatable}
 \end{theorem}

In particular, \cref{thm:vc-density} implies that
the VC-density of any fixed formula 
$\phi(\bar x,\bar y)$ over any nowhere dense class of graphs is $|\bar x|$, the number of object variables in $\phi$.

To see that the bounds provided by \cref{thm:vc-density} cannot be improved, consider a formula $\phi(\bar x,y)$ (i.e. with one parameter variable) expressing that $y$ is equal to one of the entries of 
$\bar x$. Then for each graph $G$ and parameter set $A$, $S^{\phi}(G/A)$ consists of all subsets of $A$ of size at most~$|\tup x|$, whose number is $\Theta(|A|^{|\tup x|})$. Note that
this lower bound applies to any infinite class of graphs, even edgeless ones.

We moreover show that, as long as we consider only subgraph-closed graph classes, the result of \cref{thm:vc-density} also cannot be improved in terms of generality.
The following result is an easy corollary of known characterizations of obstructions to being nowhere dense, respectively having bounded expansion.

 \newcounter{vclower}
 \setcounter{vclower}{\thetheorem}
   \begin{theorem}\label{thm:vc-density-lower-bound}
  %\begin{restatable}{theorem}{vclower}\label{thm:vc-density-lower-bound}
  Let $\CCC$ be a class of graphs which 
  is closed under taking subgraphs. 
  \begin{enumerate}[(1)]
  \item If $\CCC$ is not nowhere dense, then there is a formula 
  $\phi(x,y)$ such that for every $n\in \N$ there are $G\in\CCC$ and $A\subseteq V(G)$ 
  with $|A|=n$ and $|S^\phi(G/A)|=2^{|A|}$. 
  \item If $\CCC$ has unbounded expansion, then there is a formula 
  $\phi(x,y)$ such that for every $c\in \mathbb{R}$ there exist $G\in\CCC$ and a nonempty $A\subseteq V(G)$ with $|S^\phi(G/A)|>c|A|$. 
  \end{enumerate}
  %\end{restatable}  
\end{theorem}

\paragraph{Neighborhood complexity.}
To illustrate \cref{thm:vc-density}, consider the case when
$G$ is a graph and  $\phi(x,y)$ is the formula with two variables $x$ and $y$ expressing that the distance between $x$ and $y$
is at most $r$, for some fixed integer $r$. In this case, $S^\phi(G/A)$ is the family consisting of all intersections $U\cap A$, for $U$ ranging over all balls of radius $r$ in $G$,
and  $|S^\phi(G/A)|$ is called the \emph{$r$-neighborhood complexity} of $A$.
The concept of $r$-neighborhood complexity in sparse graph classes has already been studied before.
In particular, it was proved by Reidl et al.~\cite{reidl2016characterising} that in any graph class of bounded expansion, the $r$-neighborhood complexity of any set of vertices $A$ is $\Oof(|A|)$.
Recently, Eickmeyer et al.~\cite{eickmeyer2016neighborhood} generalized this result to an upper bound of $\Oof(|A|^{1+\epsilon})$ in any nowhere dense class of graphs.
Note that these results are special cases of \cref{thm:vc-density}.

The study of $r$-neighborhood complexity on classes of bounded expansion and nowhere dense classes was motivated by algorithmic questions from the field of parameterized complexity.
More precisely, the usage of this notion was crucial for the development of a linear kernel for the {\sc{$r$-Dominating Set}} problem on any class of bounded expansion~\cite{drange2016kernelization},
and of an almost linear kernel for this problem on any nowhere dense class~\cite{eickmeyer2016neighborhood}.
We will use the results of~\cite{drange2016kernelization,eickmeyer2016neighborhood,reidl2016characterising} on $r$-neighborhood complexity in sparse graphs in our proof of \cref{thm:vc-density}.

\paragraph{Uniform quasi-wideness.}
One of the main tools used in our proof 
is the notion of \emph{uniform quasi-wideness},
introduced by Dawar~\cite{dawar2010homomorphism}
in the context of homomorphism preservation theorems.

Formally, a class of graphs $\CCC$  is \emph{uniformly quasi-wide} if for each integer $r\in\N$ there is a function
 $N\from \N\rightarrow \N$ and a constant  $s\in \N$ such
that for every $m\in \N$, graph $G\in \CCC$, and vertex subset $A\subseteq V(G)$ of size $\abs{A}\geq N(m)$,
there is a set $S\subseteq V(G)$ of size $\abs{S}\leq s$ and a set
$B\subseteq A\setminus S$ of size $\abs{B}\geq m$ which is $r$-independent in
$G-S$. Recall that a set $B\subseteq V(G)$ is {\em{$r$-independent}} in $G$ if all
distinct $u,v\in B$ are at distance 
larger than $r$ in $G$.

Ne\v{s}et\v{r}il and Ossona de Mendez proved that
the notions of uniform quasi-wideness and nowhere denseness coincide for 
classes of finite graphs~\cite{nevsetvril2011nowhere}. 
The proof of Ne\v{s}et\v{r}il 
and Ossona de Mendez goes back to a construction
of Kreidler and Seese~\cite{kreidler1998monadic} (see also Atserias et al.~\cite{atserias2006preservation}), 
and uses iterated Ramsey arguments. Hence the original bounds on 
the function $N_r$ are non-elementary. Recently, Kreutzer, Rabinovich and the second author
 proved that for each radius $r$, we may always choose the function~$N_r$ to be a polynomial~\cite{siebertz2016polynomial}. 
 However, the exact dependence of the degree of the polynomial on $r$ and on the class $\CCC$ itself
 was not specified in~\cite{siebertz2016polynomial}, as the existence of a polynomial bound is derived
from non-constructive arguments used by Adler and Adler in~\cite{adler2014interpreting} when showing that every nowhere dense class of graphs
is stable. We remark that polynomial bounds for uniform quasi-wideness are essential for some of its applications:
the fact that $N_r$ can be chosen to be polynomial was crucially used by Eickmeyer et al.~\cite{eickmeyer2016neighborhood} both to establish an almost linear upper bound on the
$r$-neighborhood complexity in nowhere dense classes, and to develop an almost linear kernel for the {\sc{$r$-Dominating Set}} problem.
We use this fact in our proof of \cref{thm:vc-density} as well.

In our quest for constructive arguments, we give a new construction giving polynomial bounds for uniform quasi-wideness.
The new proof is considerably simpler than that of~\cite{siebertz2016polynomial}
and gives explicit and computable bounds on the degree of the polynomial.
More precisely, we prove the following theorem; here, the notation $\Oof_{r,t}(\cdot)$ hides computable factors depending on $r$ and $t$.

\newcounter{uqw}
\setcounter{uqw}{\thetheorem}
\begin{theorem}\label{thm:new-uqw}
% \newcounter{uqw}
% \setcounter{uqw}{\thetheorem}
% \begin{theorem}\label{thm:new-uqw}
%\begin{restatable}{theorem}{puqwthm}\label{thm:new-uqw}
For all $r,t\in \N$ there is a polynomial  $N\colon \N\to \N$ with $N(m)=
\Oof_{r,t}{(m^{{(4t+1)}^{2rt}})}$, such that the following holds.
Let $G$ be a graph such that $K_t\not\minor_{\lfloor 9r/2\rfloor} G$, and
let $A\subseteq V(G)$ be a vertex subset of size at least $N(m)$, for a given $m$.
Then there exists a set $S\subseteq V(G)$ of size $|S|<t$ and a set $B\subseteq A\setminus S$ 
of size $|B|\geq m$ which is $r$-independent in $G-S$.
Moreover, given~$G$ and $A$, such sets $S$ and $B$ can be computed in time $\Oof_{r,t}(|A|\cdot |E(G)|)$. 
\end{theorem}

We remark
that even though the techniques employed to prove \cref{thm:new-uqw} are inspired by methods from stability theory, 
at the end we conduct an elementary graph theoretic reasoning. In particular, as asserted in the statement, the
proof be turned into an efficient algorithm.

We also prove a result extending~\cref{thm:new-uqw}
to the case where $A\subset V(G)^d$ is a set of \emph{tuples} of vertices, of any fixed length $d$.
This result is essentially an adaptation of an analogous result due to Podewski and Ziegler~\cite{podewski1978stable} in the infinite case,
but appears to be new in the context of finite structures.
This more general result turns out to be necessary in the proof of \cref{thm:vc-density}.

\paragraph{Local separation.}
A simple, albeit important notion which permeates our proofs
is a graph theoretic concept of \emph{local separation}.
Let $G$ be a graph, $S\subset V(G)$ a set of vertices,
and let $r\in \N$ be a number. We say that two  sets of vertices $A$ and $B$  are \emph{$r$-separated} by $S$ (in $G$) if every path from a vertex in $A$ to a vertex in $B$
of length at most $r$ contains a vertex from $S$ (cf.~Fig.~\ref{fig:sep}).
 \begin{figure}[h!]
 	\centering
 		\includegraphics[scale=0.3,page=1]{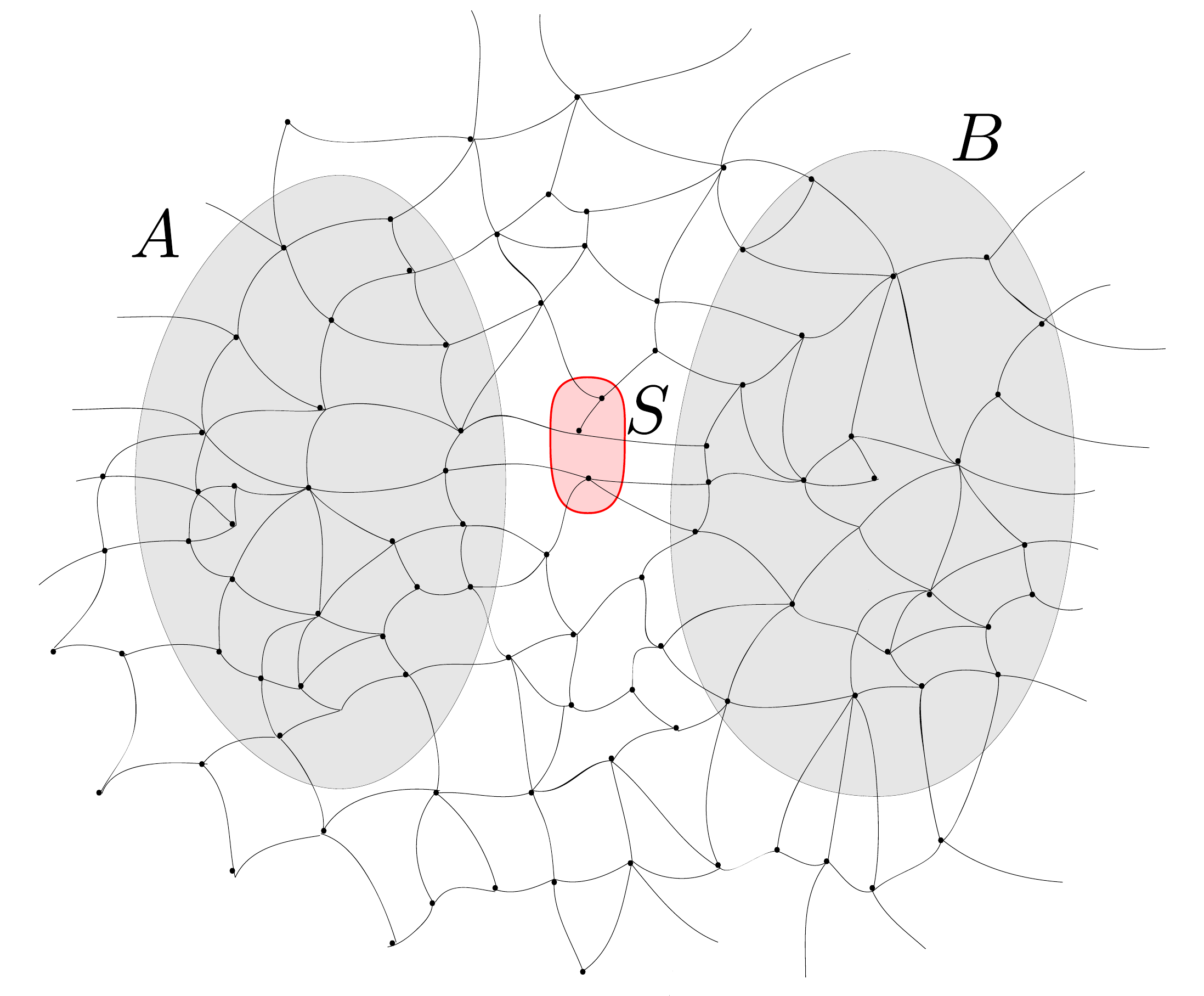}
 	\caption{The sets $A$ and $B$ are $2$-separated by $S$.
 	}
 	\label{fig:sep}
 \end{figure}
Observe that taking $r=\infty$ in $r$-separation yields the familiar notion of a separation in graph theory.
From the perspective of stability, separation (for $r=\infty$) characterizes \emph{forking independence} in superflat graphs~\cite{ivanov}. Therefore,
$r$-separation can be thought of as a local analogue of forking independence, for nowhere dense graph classes.

A key lemma concerning $r$-separation (cf.~\cref{cor:bound}) states that if $A$
and $B$ are $r$-separated by a set $S$ of size $s$ in $G$,
then for any fixed formula $\phi(\bar x,\bar y)$
of quantifier rank $\Oof(\log r)$,
the set  $\{\{\tup v\ \in B^{|\bar y|} : G\models\phi(\tup u,\tup v)\} : \tup u\in A^{|\bar x|}\}$ has cardinality bounded by a constant depending on $s$ and $\phi$ only (and not on $G,A,$ and $B$). 
This elementary result combines Gaifman's locality of first order logic (cf.~\cite{gaifman1982local}) and a Feferman-Vaught compositionality argument. This, in combination with the polynomial bounds 
for uniform quasi-wideness (\cref{thm:new-uqw}, and its extension to tuples~\cref{thm:uqw-tuples}), 
as well as the previous results on neighborhood complexity~\cite{drange2016kernelization,eickmeyer2016neighborhood}, are the main ingredients of our main result,~\cref{thm:vc-density}.

\paragraph{A duality theorem.}
As an example application of our main result,~\cref{thm:vc-density}, we prove the following
 result.
 
\newcounter{ep}
\setcounter{ep}{\thetheorem}
\begin{theorem}\label{thm:erdos-posa}
%\begin{restatable}{theorem}{erdosposa}\label{thm:erdos-posa}
	Fix a nowhere dense class of graphs $\CCC$ and a 
	formula $\phi(x,y)$ with two free variables $x,y$.
	Then there is a function $f\from \N\to\N$ with the following property.
	Let $G\in \CCC$ be a graph and let $\cal G$
	be a family of subsets of $V(G)$ consisting of sets of the form $\setof{v\in V(G)}{\phi(u, v)}$, where~$u$ is some vertex of $V(G)$.
Then~$\tau({\cal G})\le f(\nu(\cal G))$.
%\end{restatable}
\end{theorem}
\noindent Above, $\tau(\cal G)$ denotes the \emph{transversality} of $\cal G$, i.e., the least number of elements of a set $X$ which intersects every set in $\cal G$,
and $\nu(\cal G)$ denotes the \emph{packing number} of $\cal G$, i.e., the largest number of pairwise-disjoint subsets of $\cal G$.~\cref{thm:erdos-posa} is an immediate consequence of the bound given by~\cref{thm:vc-density} and a result of Matou{\v s}ek~\cite{Matousek:2004:BVI:1005787.1005789}.

We remark that a similar, but incomparable result
is proved by Bousquet and Thomass{\'e}~\cite{BousquetT15}.
In their result, the assumption on $\CCC$ is weaker, since they just require that it has \emph{bounded distance VC-dimension}, 
but the assumption on   $\cal G$ is stronger, as it is required to be the set of all balls of a fixed radius.

\paragraph{Stability.}
Finally, we observe that we can apply our  tools to give a constructive proof of the result of Adler and Adler~\cite{adler2014interpreting}
that every nowhere dense class is stable, which yields computable upper bounds on ladder indices.
More precisely, we translate the approach of Podewski and Ziegler~\cite{podewski1978stable} to the finite
and replace the key non-constructive application of compactness with a combinatorial argument based on Gaifman's locality,
in the flavor served by our observations on $r$-separation (\cref{cor:bound}).
The following theorem summarizes our result.

 \newcounter{stable}
 \setcounter{stable}{\thetheorem}
%\begin{restatable}{theorem}{newstable}\label{thm:new-stable}
 \begin{theorem}\label{thm:new-stable}
There are computable functions $f\colon \N^3\to\N$ and $g\colon\N\to\N$ with the following property.
Suppose $\phi(\bar x,\bar y)$ is a formula of quantifier rank at most $q$ and with $d$ free variables.
Suppose further that $G$ is a graph excluding $K_t$ as a depth-$g(q)$ minor. Then the ladder index of $\phi(\bar x,\bar y)$ in $G$ is at most $f(q,d,t)$.
 \end{theorem}
%\end{restatable}

%Note that in particular, \cref{thm:new-stable} implies that every nowhere dense graph is stable, which was the main conclusion of the paper by Adler and Adler~\cite{adler2014interpreting}. 

\pagebreak

\paragraph{Organization.} In \cref{sec:prelim} we recall some standard concepts from the theory of sparse graphs.
In~\cref{sec:uqw} we  prove \cref{thm:new-uqw}, improving the previously known bounds and making them constructive. We remark that this result is not needed in the proof of our main result,~\cref{thm:vc-density}. The following two sections
 contain the main tools needed in the proof of the main result:
in~\cref{sec:uqw-tuples} we formulate and prove the generalization of uniform quasi-wideness to tuples,~\cref{thm:uqw-tuples}, and 
%This result is new in the context of nowhere dense graph classes, and is an important tool for the further results.
in \cref{sec:gaifman} we discuss Gaifman locality for first order logic and derive an elementary variant concerning local separators. 
In \cref{sec:types} we prove our main result, \cref{thm:vc-density}, and the corresponding lower bounds, \cref{thm:vc-density-lower-bound}.
Finally, in \cref{sec:stable} we provide a constructive proof of the result of Adler and Adler, \cref{thm:new-stable}.

\paragraph{Acknowledgments.} We would like to
thank Patrice Ossona de Mendez for pointing us to the
question of studying VC-density of nowhere dense graph
classes.
%\pagebreak
% \todo{page break here?}

\section{Preliminaries}\label{sec:prelim}
We  recall some basic notions from graph theory. 

\medskip
All graphs in this paper are finite, undirected and simple, that is, 
they do not have loops or parallel edges. Our notation is standard,
we refer to~\cite{diestel2012graph} for more background on 
graph theory. 
We write $V(G)$ for the vertex set of a graph $G$ and
$E(G)$ for its edge set. 
The {\em{distance}} between vertices $u$ and $v$ in $G$, denoted $\dist_G(u,v)$, is the length of a shortest path between $u$ and $v$ in~$G$.
If there is no path between $u$ and $v$ in $G$, we put $\dist_G(u,v)=\infty$.
The {\em{(open) neighborhood}} of a vertex $u$, denoted $N(u)$, is the set of neighbors of $u$, excluding $u$ itself.
For a non-negative integer $r$, by $N_r[u]$ we denote the {\em{(closed) $r$-neighborhood}} of $u$ which comprises vertices at distance at most $r$ from $u$; 
note that $u$ is always contained in its closed $r$-neighborhood. The \emph{radius} of a connected graph $G$ is the least integer $r$ such that there is some vertex $v$ of $G$ with $N_r[v]=V(G)$.

A {\em{minor model}} of a graph $H$ in $G$ is a family $(I_u)_{u\in V(H)}$ of pairwise vertex-disjoint connected subgraphs of $G$, called {\em{branch sets}},
such that whenever $uv$ is an edge in~$H$, there are $u'\in V(I_u)$ and $v'\in V(I_v)$ for which $u'v'$ 
is an edge in $G$.
The graph $H$ is a {\em{depth-$r$ minor}} of $G$, denoted $H\minor_rG$, if there is a minor model
$(I_u)_{u\in V(H)}$ of~$H$ in $G$ such that each $I_u$ has radius at most $r$.

A class $\CCC$ of graphs is \emph{nowhere dense} if there is a function 
$t\colon \N\rightarrow \N$ such that for all $r\in \N$ it holds that $K_{t(r)}\not\minor_r G$
for all $G\in \CCC$, where $K_{t(r)}$ denotes the clique on $t(r)$ vertices.
The class~$\CCC$ moreover has \emph{bounded expansion}
if there is a function $d\colon\N\rightarrow\N$ such that for all 
$r\in \N$ and all $H\minor_rG$ with $G\in\CCC$, the {\em{edge density}}
of $H$, i.e. $|E(H)|/|V(H)|$, is bounded by $d(r)$. Note that every 
class of bounded expansion is nowhere dense. The converse is not necessarily true in general~\cite{sparsity}.

% A set $B\subseteq V(G)$ is called {\em{$r$-independent}} in a graph $G$ if  $\dist_G(u,v)>r$ for all
% distinct $u,v\in B$.

A set $B\subseteq V(G)$ is called {\em{$r$-independent}} in a graph $G$ if  $\dist_G(u,v)>r$ for all
distinct $u,v\in B$. 
A class $\CCC$ of graphs is \emph{uniformly quasi-wide} if for every $r\in \N$ there is a number $s\in \N$
and a function $N\from\N\to\N$ such that
for every $m\in \N$, graph $G\in \CCC$, and vertex subset $A\subseteq V(G)$ of size $\abs{A}\geq N(m)$, there is a set
$S\subseteq V(G)$ of size $\abs{S}\leq s$ and a set
$B\subseteq A-S$ of size $\abs{B}\geq m$ 
which is $r$-independent in $G-S$.
Recall that Ne\v set\v ril and Ossona de 
Mendez proved~\cite{nevsetvril2011nowhere} that nowhere dense graph classes are exactly the same as uniformly quasi-wide classes. 
The following result of Kreutzer, Rabinovich and the second author~\cite{siebertz2016polynomial}
improves their result, by showing that the function $N$ can be taken polynomial:

\begin{theorem}[\cite{siebertz2016polynomial}]\label{thm:krs}
For every nowhere dense class $\CCC$
and for all $r\in \N$ there is a polynomial $N\from \N\to\N$ 
and a number $s\in \N$ such that the following holds.
Let be $G\in \CCC$ be aa graph and let $A\subset V(G)$ be a vertex subset of size at least $N(m)$, for a given $m$.
Then there exists a set $S\subset V(G)$ of size $|S|<s$
and a set $B\subset A-S$ of size $|B|\ge m$ which is $r$-independent in $G-S$.	
\end{theorem}

As we mentioned, the proof of Kreutzer et al.~\cite{siebertz2016polynomial} relies on non-constructive arguments and does not yield explicit bounds on $s$ and (the degree of) $N$. 
In the next section, we discuss a further strengthening of this result, by providing explicit, computable bounds on $N$ and $s$.
\section{Bounds for uniform quasi-wideness}\label{sec:uqw}

In this section we prove \cref{thm:new-uqw}, which strengthens \cref{thm:krs} by providing an explicit polynomial $N$ and bound $s$,
whereas the bounds in~\cref{thm:krs} rely on non-constructive arguments. 
We note that~\cref{thm:krs} is sufficient to prove our main result,~\cref{thm:vc-density}, but is required in our proof of~\cref{thm:new-stable}, which is the effective 
variant of the result of Adler and Adler, \cref{thm:new-stable}.

%
% \setcounter{theorem}{}
% \begin{theorem}\label{thm:new-uqw}
% For all $r,t\in \N$ there is a polynomial  $N\colon \N\to \N$ with $N(m)=
% \Oof_{r,t}{(m^{{(4t+1)}^{2rt}})}$, such that the following holds.
% Let $G$ be a graph such that $K_t\not\minor_{\lfloor 9r/2\rfloor} G$, and
% let $A\subseteq V(G)$ be a vertex subset of size at least $N(m)$, for a given $m$.
% Then there exists a set $S\subseteq V(G)$ of size $|S|<t$ and a set $B\subseteq A\setminus S$
% of size $|B|\geq m$ which is $r$-independent in $G-S$.
% Moreover, given~$G$ and $A$, such sets $S$ and $B$ can be computed in time $\Oof_{r,t}(|A|\cdot |E(G)|)$.
% \end{theorem}

\paragraph{General strategy.}
Our proof follows the same lines as the original proof of Ne\v set\v ril and Ossona de Mendez~\cite{nevsetvril2011nowhere}, with the difference that in the key technical lemma (\cref{lem:apex} below), 
we improve the bounds significantly by replacing a Ramsey argument with a refined combinatorial analysis.
The new argument essentially originates in the concept of {\em{branching index}} from stability theory. 
%, 
%and also uses the almost linear bound on neighborhood complexity in nowhere dense graph classes, due to Gajarsk\'y et al.~\cite{gajarsky2017kernelization}. \marginpar{no longer}
%For sake of completeness, we present the entire proof of~\cref{thm:new-uqw}.

We first prove a restricted variant,~\cref{lem:engine} below, in which we assume that $A$ is already $(r-1)$-independent. Then, in order to derive
\cref{thm:new-uqw}, we apply the lemma iteratively for $r$ ranging from $1$ to the target value.

\begin{lemma}\label{lem:engine}
For every pair of integers $t,r\in \N$ there exists an integer $d<9r/2$ and a function $L\colon \N\to \N$ with $L(m)=\Oof_{r,t}{(m^{{(4t+1)}^{2rt}})}$ such that the following holds.
For each $m\in \N$, graph~$G$ with $K_t\not\minor_{d} G$, and
$(r-1)$-independent set $A\subseteq V(G)$ of size at least $L(m)$, there is a set $S\subseteq V(G)-A$ of size less than $t$ such that $A$ contains a subset $B$ of size $m$ which is $r$-independent in $G-S$.
Moreover, if $r$ is odd then $S$ is empty, and if $r$ is even,
then every vertex of $S$ is at distance exactly $r/2$ from every vertex of $B$.
Finally, given $G$ and $A$, the sets $B$ and $S$ can be computed in time $\Oof_{r,t}(|A|\cdot |E(G)|)$.
\end{lemma}

We prove~\cref{lem:engine} in \Cref{sec:engine}, but  a very rough sketch is as follows.
The  case of general~$r$ reduces to the case $r=1$ or $r=2$, depending on the parity of $r$,
by contracting the balls of radius $\lfloor \frac {r-1} 2\rfloor $ around the vertices in $A$ to single vertices.
The case of $r=1$ follows immediately from Ramsey's theorem, as in~\cite{nevsetvril2011nowhere}.
The case $r=2$ is substantially more difficult.
We start by formulating and proving the main technical result needed for proving the case $r=2$.

\subsection{The main technical lemma}
\label{sec:main-tech}

The following, Ramsey-like result is the main technical lemma used in the proof of~\cref{thm:new-uqw}. 

\pagebreak

\begin{lemma}\label{lem:apex}
Let $\ell,m,t\in \N$ and assume $\ell\geq t^{8}$. 
If~$G$ is a graph and $A$ is a $1$-independent set in~$G$
with at least $(m+\ell)^{2t}$ elements,
then at least one of the following conditions hold:
\begin{itemize}
  \item $K_t\minor_{4} G$,
\item  $A$ contains a $2$-independent set of size $m$, 
\item  some vertex $v$ of $G$
has at least $\ell^{1/4}$ neighbors in $A$.
\end{itemize}
Moreover, if $K_t\not\minor_4G$, the
structures described in the other two cases (a $2$-independent set 
of size~$m$, or a vertex $v$ as above) can be 
computed in time $\Oof_t(|A|\cdot |E(G)|)$. 
\end{lemma}
We remark that a statement similar to that of \cref{lem:apex}
can be obtained by employing Ramsey's theorem, as has been done in~\cite{nevsetvril2011nowhere}. This, however, 
%yields
%in place of the bound $(m+\ell)^{2t}$ 
%a bound of the form $R(m,\underbrace{q,q,\ldots,q}_{k\text{ times}})$,
%where $k\sim\ell^{1/8}$ and $R(m_1,\ldots,m_c)$
%is the Ramsey number for $c$ colors.
%In particular, this 
does not give a bound which is polynomial in $m+\ell$, and thus cannot be used to prove~\cref{thm:new-uqw}.

\medskip
The remainder of this section is devoted to the proof of~\cref{lem:apex}.
We will use the following bounds on the edge density
of graphs with excluded shallow minors obtained
by Alon et al.~\cite{alon2003turan}. 

\begin{lemma}[Theorem 2.2 in~\cite{alon2003turan}]\label{lem:densitynd}
Let $H$ be a bipartite graph with maximum degree
$d$ on one side. Then there exists a constant $c_H$, depending 
only on $H$, such that every $n$-vertex graph $G$
excluding~$H$ as a subgraph has at most $c_H\cdot n^{2-1/d}$
edges. 
\end{lemma} 

Observe that if $K_t\not\minor_1G$, then in particular
the $1$-subdivision of $K_t$ is excluded as a subgraph
of $G$ (the $1$-subdivision of a graph $H$ is obtained by 
replacing every edge of $H$ by a path of length $2$). 
Moreover, the $1$-subdivision of 
$K_t$ is a bipartite graph with maximum degree $2$ on one
side. Furthermore, it is easy to check in the 
proof of Theorem 2.2 in~\cite{alon2003turan} 
that $c_H\leq |V(H)|$
in case $d=2$. Since the $1$-subdivision of $K_t$ has 
$\binom{t+1}{2}$ vertices, we can choose $c_{K_t}=\binom{t+1}{2}$ and
conclude the following.   

\begin{corollary}\label{crl:densitynd}
Let $G$ be an $n$-vertex graph such that $K_t\not\minor_1 G$ for
some constant $t\in \N$. Then $G$ has at most
$\binom{t+1}{2}\cdot n^{3/2}$ edges.
\end{corollary}

We will use the following standard lemma saying that a shallow minor of a shallow minor is a shallow minor, where the parameters of shallowness are appropriately chosen.

\begin{lemma}[adaptation of Proposition 4.11 in~\cite{sparsity}]\label{lem:combineminors}
Suppose $J,H,G$ are graphs such that $H\minor_a G$ and $J\minor_b H$, for some $a,b\in \N$.
Then $J\minor_c G$, where $c=2ab+a+b$.
\end{lemma}

We will need one more technical lemma.

\begin{lemma}\label{lem:diversity}
  Let $G$ be a graph such that $K_t\not\minor_4G$ for some
  $t\in\N$ and let $A\subseteq V(G)$ with $|A|\geq t^{8}$. 
  Assume furthermore that every pair of elements of $A$ has a common neighbor in $V(G)\setminus A$.
  Then there exists a vertex $v$ in $V(G)\setminus A$ which has at least $|A|^{1/4}$ neighbors in $A$.
\end{lemma}
\begin{proof}
Denote $k=\max\{\,|N(w)\cap A|\ \colon\ w\in V(G)-A\,\}$; our goal is to prove that $k\geq |A|^{1/4}$.

Let $B\subseteq V(G)-A$ be the set of those vertices outside of $A$ that have a  neighbor in $A$. 
Construct a function $f\colon B\to A$ by a random procedure as follows:
for each vertex $v\in B$, choose $f(v)$ uniformly and independently at random from the set $N(v)\cap A$.
Next, for each $u\in A$ define branch set $I_u=G[\{u\}\cup f^{-1}(u)]$. Observe that since, by construction, $v$ and $f(v)$ are adjacent for all $v\in B$, each branch set $I_u$ has radius at most $1$,
with $u$ being the central vertex. Also, the branch sets $\{I_u\}_{u\in A}$ are pairwise disjoint.
Finally, construct a graph $H$ on vertex set $A$ by making distinct $u,v\in A$ adjacent in $H$ whenever there is an edge in $G$ between the branch sets~$I_u$ and $I_v$.
Then the branch sets $\{I_u\}_{u\in A}$ witness that $H$ is a $1$-shallow minor of $G$.

For distinct $u,v\in A$, let us estimate the probability that the edge $uv$ appears in $H$.
By assumption, there is a vertex $w\in B$ that is adjacent both to $u$ and to $v$. Observe that if it happens that $f(w)=u$ or $f(w)=v$, then $uv$ for sure becomes an edge in $H$. 
Since $w$ has at most $k$ neighbors in $A$, the probability that $f(w)\in \{u,v\}$ is at least $\frac{2}{k}$.

By the linearity of expectation, the expected number of edges in $H$ is at least $\binom{|A|}{2}\cdot \frac{2}{k}=\frac{|A|(|A|-1)}{k}$.
Hence, for at least one run of the random experiment we have that $H$ indeed has at least this many edges. 
On the other hand, observe that $K_t\not\minor_1 H$; indeed, since $H\minor_1 G$, by Lemma~\ref{lem:combineminors} we infer that $K_t\minor_1 H$ would imply $K_t\minor_4 G$, a contradiction with the assumptions on $G$.
Then Corollary~\ref{crl:densitynd} implies $H$ has at most $\binom{t+1}{2}\cdot |A|^{3/2}$ edges.
Observe that 
$\binom{t+1}{2}\cdot |A|^{3/2}\leq 3t^2/4\cdot |A|^{3/2}\leq \frac{3}{4}|A|^{7/4}$,
where the first inequality holds due to $t\geq 2$, while the second holds by the assumption that $|A|\geq t^8$.
By combining the above bounds, we obtain
$$\frac{|A|(|A|-1)}{k}\leq \frac{3}{4}|A|^{7/4},$$
which implies $k\geq |A|^{1/4}$ due to $|A|\geq t^8\geq 64$.
\end{proof}

We proceed with the proof of \cref{lem:apex}.
The idea is to arrange the elements of $A$ in a binary tree
and prove that provided $A$ is large, this tree contains a long path. From this path, we will 
extract the set $B$. 
In stability theory, similar trees are called \emph{type trees} and they are used to extract long indiscernible sequences, see e.g.~\cite{malliaris2014regularity}.

\newcommand{\dau}{\mathrm{D}}
\newcommand{\son}{\mathrm{S}}
	
	We will work with a two-symbol alphabet $\set{\dau,\son}$, for {\em{daughter}} and {\em{son}}.
	We identify words in $\set{\dau,\son}^*$ with \emph{nodes}
	of the infinite rooted binary tree. 
  The \emph{depth} of a node $w$ is the length of $w$.
  For $w\in \set{\dau,\son}^*$,
	 the nodes $w\dau$ and $w\son$ are called, respectively, the \emph{daughter} and the \emph{son} of $w$,
	and $w$ is the \emph{parent} of both $w\son$ and $w\dau$. A node $w'$ is a {\em{descendant}} of a node $w$ if $w'$ is a prefix of $w$ (possibly $w'=w$).
	We consider
	 finite, labeled, rooted, binary trees, which are called simply trees below, and are defined as follows.
	 For a set of labels $U$, a ($U$-labeled) \emph{tree} is a partial function $\tau\from \set{\dau,\son}^*\to U$ whose domain is a finite set of nodes, 
	 called the \emph{nodes of $\tau$}, which is closed under taking parents. 
	 If $v$ is a node of $\tau$, then $\tau(v)$ is called its \emph{label}.
  
  Let $G$ be a graph, $A\subset V(G)$ be a $1$-independent set in $G$,
  and $\bar a$ be any enumeration of $A$, that is, a sequence of length $|A|$ in which every element of $A$ appears exactly once.
  We define a binary tree $\tau$ which is 
  labeled by vertices of $G$. The tree is defined by processing all elements of~$\bar a$ sequentially. 
  We start with $\tau$ being the  tree with empty domain, and for each element $a$ of the sequence $\bar a$, processed in the order given by $\bar a$, 
  execute the following procedure which results in adding a node with label $a$ to $\tau$.
  
When processing the vertex $a$, do the following. Start with $w$ being the empty word. While~$w$ is a node of $\tau$, repeat the following step: 
  if the distance from $a$ to $\tau(w)$ in the graph 
  $G$ is at most~$2$, replace $w$ by its son, otherwise, replace $w$ by its daughter.
  % Repeat the step, unless $\tau(w)$ is undefined.
  Once $w$ is not a node of $\tau$, extend $\tau$ by setting  $\tau(w)=a$. In this way, we have processed the element $a$, and now
    proceed to the next element of $\bar a$, until all elements are processed. This ends the construction of $\tau$.
    Thus, $\tau$ is a tree labeled with vertices of $A$, and every vertex of $A$ appears exactly once in $\tau$.

Define the
\emph{depth} of $\tau$ as 
the maximal depth of a node of $\tau$.
For a word $w$, an \emph{alternation} in~$w$ is any 
position $\alpha$, $1\leq \alpha\leq |w|$, such that $w_\alpha\neq w_{\alpha-1}$; here, $w_\alpha$ denotes the $\alpha$th symbol of~$w$, and~$w_0$ is assumed to be $\dau$.
The \emph{alternation rank} of the tree $\tau$ is the maximum of the number of alternations in $w$, over all nodes $w$ of $\tau$.

\begin{lemma}\label{lem:number-of-nodes}
Let $h,t\ge 2$.	If $\tau$ has alternation rank at most $2t-1$ and depth at most $h-1$, then~$\tau$ has fewer than $h^{2t}$ nodes.
\end{lemma}
\begin{proof}		
	With each node $w$ of $\tau$ associate
	function $f_w\colon \set{1,\ldots,2t}\to\set{1,\ldots,h}$ defined as follows:
	$f_w$ maps each $i\in \set{1,\ldots,2t}$ to the $i$th alternation of $w$, provided $i$ is at most the number of alternations of $w$, and otherwise we put $f_w(i)=|w|+1$.
	It is clear that the mapping $w\mapsto f_w$ for nodes $w$ of $\tau$ is injective
and its image is contained in monotone functions from $\set{1,\ldots,2t}$ to $\set{1,\ldots,h}$, whose number is less than $h^{2t}$.
Hence, the domain of~$\tau$ 
		has fewer than $h^{2t}$ elements.
\end{proof}

\begin{lemma}\label{thm:alternation-rank-type-tree}
Suppose that  $K_t\not\minor_{2} G$.
Then $\tau$ has alternation rank at most $2t-1$.
\end{lemma}
\begin{proof}
	Let $w$ be a node of $\tau$ with at least $2k$ alternations, for some $k\in \N$.
	Suppose $\alpha_1,\beta_1,\ldots,\alpha_k,\beta_k$ be the first $2k$ alternations of $w$.
	By the assumption that $w_0=\dau$ we have that~$w$ contains symbol $\son$ at all positions $\alpha_i$ for $i=1,\ldots,k$, and symbol $\dau$ at all positions $\beta_i$ for $i=1,\ldots,k$.
	For each $i\in \set{1,\ldots,k}$, define $a_i\in V(G)$ to be the label in $\tau$ of the prefix of $w$ of length $\alpha_i-1$, and similarly define $b_i\in V(G)$ to be the label in $\tau$ of the prefix of $w$
	of length $\beta_i-1$. 
	It follows that for each $i\in \set{1,\ldots,k}$, the following assertions hold:
	the nodes in $\tau$ with labels $b_i,a_{i+1},b_{i+1},\ldots,a_k,b_k$ are  descendants of the son of the node with label $a_i$,
	and the nodes with labels $a_{i+1},b_{i+1},\ldots,a_k,b_k$
	are descendants of the daughter of the node with label $b_i$.
	
	\begin{claim}\label{claim:minor}
		For every pair $a_i,b_j$ with $1\le i\le j\le k$, there is a vertex $z_{ij}\not\in A$	 which is a common neighbor of $a_i$ and $b_j$,
		and is not a neighbor of any $b_s$ with $s\neq j$.
	\end{claim}
	\begin{clproof}
		Note that since $i\le j$, the node with label $b_j$ is a descendant of the son of the node with label $a_i$, hence we have $\dist_G(a_i,b_j)\le 2$ by the construction of $\tau$.
		However, we also have $\dist_G(a_i,b_j)>1$ since $A$
		is $1$-independent. Therefore $\dist_G(a_i,b_j)=2$, so there is a vertex $z_{ij}$ which is a common neighbor of $a_i$ and $b_j$. 
		Suppose that $z_{ij}$ was a neighbor of $b_s$, for some $s\neq j$. This would imply that $\dist_G(b_j,b_s)\le 2$, which is impossible, 
because
		 the nodes with labels~$b_s$ and $b_j$ in $\tau$ are such that one is a descendant of the daughter of the other, implying that $\dist_G(b_s,b_j)>2$.
	\end{clproof}
  
Note that whenever $i\leq j$ and $i'\leq j'$ are such that $j\neq j'$, the vertices $z_{ij}$ and $z_{i'j'}$ are different, because $z_{ij}$ is adjacent to $b_{j}$ but not to $b_{j'}$, and the converse holds for $z_{i'j'}$.
However, it may happen that $z_{ij}=z_{i'j}$ even if $i\neq i'$. This will not affect our further reasoning.

For each $j\in\set{1,\ldots,k}$, let $B_j$
be the subgraph of $G$ induced by the set
$\set{a_j,b_j}\cup\set{z_{ij}\colon 1\le i\le  j}$.
Observe that $B_j$ is connected and has radius at most $2$, with $b_j$ being the central vertex.
By \cref{thm:alternation-rank-type-tree} and the discussion from the previous paragraph, the graphs $B_j$ for $j\in \set{1,\ldots,k}$
are pairwise disjoint.
Moreover, for all $1\le i\le j\le k$, there is an edge between $B_i$
and $B_j$, namely, the edge between $z_{ij}\in B_j$
and $a_i\in B_i$.
Hence, the graphs $B_j$, for $j\in \set{1,\ldots,k}$, define a depth-$2$ minor model of $K_k$ in $G$. Since $K_t\not\minor_{2}G$, this implies that $k<t$, proving~\cref{thm:alternation-rank-type-tree}.
\end{proof}

We continue with the proof of~\cref{lem:apex}. 
Fix integers $\ell\ge t^8$ and~$m$, and define $h=m+\ell$.
Let $A$ be a $1$-independent set in $G$
of size at least $h^{2t}$.

Suppose that the first case of \cref{lem:apex} does not hold. In particular $K_t\not\minor_2 G$, so by \cref{thm:alternation-rank-type-tree},~$\tau$ has alternation rank at most $2t-1$. From \cref{lem:number-of-nodes} 
we conclude that $\tau$  has depth at least~$h$.
As $h=m+\ell$, it follows that either $\tau$  has a node $w$ which contains at least $m$ letters~$\dau$, or $\tau$ has a node $w$ which contains  at least $\ell$ letters $\son$.

Consider the first case, i.e., there is a node $w$ of $\tau$
which contains at least $m$ letters $\dau$, and let $X$
be the set of all vertices $\tau(u)$ such that $u\dau$ is a prefix of $w$. Then, by construction, $X$ is a $2$-independent set in $G$ of size at least $m$, so the second case of the lemma holds.

Finally, consider the second case, i.e., there is a node $w$ in $\tau$ which contains at least $\ell$ letters~$\son$. Let 
$Y$ be the set of all vertices $\tau(u)$ such that $u\son$ is a prefix of $w$. Then, by construction, $Y\subset A$ is a set of at least $\ell$ vertices which are mutually at distance exactly $2$ in $G$. 
Since $K_t\not\minor_4 G$ and $\ell\geq t^8$, by~\cref{lem:diversity} we infer that there is a vertex $v\in G$
with at least $\ell^{1/4}$ neighbors in $Y$.
This finishes the proof of the existential part of~\cref{lem:apex}.

For the algorithmic part, the proof above yields an algorithm which first constructs the tree $\tau$, by 
iteratively processing each vertex $w$ of $A$ and testing whether the distance between $w$ and each vertex processed already is equal to $2$.
This amounts to running a breadth-first search from every vertex of $A$, which can be done in time $\Oof(|A|\cdot |E(G)|)$.
Whenever a node with $2t$ alternations 
is inserted to $\tau$, we can exhibit in $G$ a depth-$2$ minor model of $K_t$.
Whenever a node with least $m$ letters $\dau$ is added to~$\tau$,
we have constructed an $m$-independent set. Whenever a node with at least $\ell$ letters $\son$ is added to $\tau$, as argued, there must be some vertex $v\in V(G)-A$ with at least $\ell^{1/8}$ neighbors in~$A$. 
To find such a vertex, scan through all neighborhoods of vertices $v\in A$ in the graph $G$, and then select a vertex $w\in V(G)$
which belongs to the largest number of those neighborhoods; this can be done in time $\Oof(|E(G)|)$.
The overall running time is $\Oof(|A|\cdot |E(G)|)$, as required.

%This finishes the proof of~\cref{lem:apex}.

\subsection{Proof of~\cref{lem:engine}}
\label{sec:engine}
% To prove~\cref{lem:engine}, we distinguish two special cases: the case of $r=1$ and the case $r=2$. The case of general $r$ then reduces to one of these two cases, depending on the parity of $r$, by observing that a $(2s+1)$-independent set $A$ in $G$
% induces a $1$-independent set in $G$ with the balls of radius $s$ around the vertices of $A$ contracted, and,
% similarly, a $(2s+2)$-independent set $A$ in $G$
% induces a $2$-independent set in $G$ with the balls of radius $s$ around the vertices of $A$ contracted.

With \cref{lem:apex} proved, we can proceed with~\cref{lem:engine}. 
We start with the case $r=1$, then we move to the case $r=2$. 
Next, we show how the general case reduces to one of those two cases.
% , and, finally, we deduce~\cref{thm:new-uqw} from~\cref{lem:engine}.

\paragraph{Case $r=1$.}
We put $d=0$, thus we assume that $K_t\not\minor_0 G$; that is, $G$ does not contain a clique of size $t$ as a subgraph. By Ramsey's Theorem, in every graph every vertex subset of size $\binom{m+t-2}{t-1}$ contains an
independent set of size $m$ or a clique of size $t$. Therefore, 
taking $L(m)$ to be the above binomial coefficient yields~\cref{lem:engine} in case $r=0$, for $S=\emptyset$. Note here that $\binom{m+t-2}{t-1}\in\Oof_{t}{(m^{{(4t+1)}^{2t}})}$.
Moreover, such independent set or clique can be computed from $G$ and $A$ in time~$\Oof(|A|\cdot |E(G)|)$ by simulating the proof of Ramsey's theorem.

\paragraph{Case $r=2$.}
We put $d=2$, thus we assume that $K_t\not\minor_4 G$.
We show that if $A$ is a sufficiently large $1$-independent set in a graph $G$ such that $K_t\not\minor_4 G$, 
then there is a set of vertices~$S$ of size less than $t$ such that $A\setminus S$ contains a subset of size $m$ which is $2$-independent in $G-S$. 
Here, by ``sufficiently large'' we mean of size of size at least $L(m)$, for $L(m)$ emerging from the proof.
To this end, we shall iteratively apply \cref{lem:apex} as long as  it results in the third case, 
yielding a vertex $v$ with many neighbors in $A$. In this case, we add $v$ vertex to the set $S$, and apply the lemma again,
restricting $A$ to $A\cap N(v)$. 
Precise calculations follow.

\newcommand{\mbull}{\widehat{m}}

Fix a number $\beta>4t$. For $k\ge 0$,
define $m_k=((k+1)\cdot m)^{(2\beta)^k}$.
In the following we will always assume that $m\geq t^8$. 
We will apply~\cref{lem:apex} in the following form.
\begin{claim}\label{cor:apex}
	If $G$ is a graph such that $K_t\not\minor_4 G$, and
	$A\subset V(G)$ is an $1$-independent set in $G$ which does not contain a $2$-independent set of size $m$ and satisfies $|A|\ge m_k$, for some $k\geq 1$,
	then there exists a vertex $v\in V(G)-A$ such that $|N_G(v)\cap A| \ge m_{k-1}$.
\end{claim}
\begin{clproof}
Let $\ell=(k\cdot m)^{4\cdot(2\beta)^{k-1}}$.
Then $m\ge t^8$ implies that $\ell\ge t^8$.
Observe that
\[|A|\ge \left((k+1)\cdot m\right)^{(2\beta)^k}\ge\left ((m+ k\cdot m)^{4\cdot(2\beta)^{k-1}} \right)^{2t}
\ge \left(m+(k\cdot m)^{4\cdot (2\beta)^{k-1}}\right)^{2t}=(m+\ell)^{2t}.\]
Therefore, we may  apply \cref{lem:apex}, yielding a vertex $v$ with at least $\ell^{1/4}=(k\cdot m)^{(2\beta)^{k-1}}=m_{k-1}$ neighbors in~$A$.
\end{clproof}

%In the following we assume that $m\geq t^8$, since we may always ask for finding a $2$-independent set of size $t^8$ instead of $m$.
We will now find 
a subset of $A$ of size $m$ which is $2$-independent in $G-S$, for some $S$ with $|S|<t$.
Assume that $|A|\ge m_t$. By induction, we
 construct a sequence  $A=A_0\supseteq A_1\supseteq\ldots$ 
of \mbox{$1$-independent} vertex subsets of $G$
of length at most $t$
such that $|A_i|\ge m_{t-i}$,
 as follows. Start with $A_0=A$. We maintain a set $S$ of vertices of $G$ which is initially empty, and we maintain the invariant that $A_i$ is disjoint with $S$ at each step of the construction.

For $i=0,1,2,\ldots$ do as follows.
If $A_{i}$ contains a subset of size $m$ which is $2$-independent set in $G-S$, terminate.
 Otherwise, 
 apply~\cref{cor:apex} to the graph $G-S$ with $1$-independent set
 $A_{i}$ of size $|A_i|\ge m_{t-i}$. This yields a vertex $v_{i+1}\in V(G)-(S\cup A_i)$
 whose neighborhood in $G-S$ contains at least
 $m_{t-i-1}$ vertices of $A_{i}$.
 Define $A_{i+1}$ as the set of neighbors of $v_{i+1}$ in $A_i$, and add $v_{i+1}$
 to the set~$S$.  
  Increment $i$ and repeat.

\begin{claim}\label{claim:at-most-t}
	The construction halts after less than $t$ steps.
\end{claim}
\begin{clproof}
Suppose that the construction proceeds for $k\le t$ steps.
By construction, each vertex~$v_i$, for $i\le k$, is adjacent in $G$
 to all the vertices of $A_{j}$, for each $i\le j\le k$. In particular, all the vertices $v_1,\ldots,v_k$ are adjacent to all the vertices of $A_{k}$
 and $|A_k|\ge m_{t-k}\ge m\ge t$.
Choose any pairwise distinct vertices $w_1,\ldots,w_k\in A_k$ and observe that the connected subgraphs $G[\set{w_i,v_i}]$ of~$G$ yield a depth-$1$ minor model of $K_k$ in $G$.
 Since $K_t\not\minor_2 G$, we must have $k<t$.
 \end{clproof}
 
 Therefore, at some step $k<t$ of the construction we must have obtained a $2$-independent subset $B$ of $G-S$ of size $m$. Moreover, $|S|\le k<t$.

 This proves~\cref{lem:engine} in the case $r=2$, for the function $L(m)$ defined as $L(m)=m_t=((t+1)\cdot m)^{\beta^{2t}}$
 for $m\ge t^8$, and $L(m)=L(t^8)$ for $m<t^8$, where $\beta>4t$ is any fixed constant.
 It is easy to see that then $L(m)\in \Oof_{t}{(m^{{(4t+1)}^{2t}})}$, provided we put $\beta=4t+1$.
 Also, the proof easily yields an algorithm constructing the sets~$B$ and $S$,
 which amounts to applying at most $t$ times the algorithm of~\cref{lem:apex}.
 Hence, its running time  is $\Oof_{r,t}(|A|\cdot |E(G)|)$, as required.
% \end{proof}

\paragraph{Odd case.}
We now prove~\cref{lem:engine} in the case when $r=2s+1$, for some integer $s\geq 1$. We put $d=s=\frac{r-1}{2}$.
Let $G$ be a graph such that $K_t\not\minor_s G$, and 
 let $A$ be a $2s$-independent set in $G$. Consider the graph $G'$ obtained from $G$
by contracting the (pairwise disjoint) balls of radius $s$ around each vertex $v\in A$.
 Let $A'$ denote the set of vertices of $G'$ corresponding to the contracted balls. There is a natural correspondence (bijection) between $A$ and $A'$, where each vertex $v\in A$ is associated with the
 vertex of $A'$ resulting from contracting the ball of radius $s$ around $v$.
From $K_t\not\minor_s G$ it follows that~$G'$ does not contain $K_t$ as a subgraph. Applying the already proved case $r=1$ to $G'$ and $A'$, we conclude that 
provided $|A|=|A'|\ge {m+t-2\choose t-1}$, the set
 $A'$ contains a $1$-independent subset $B'$ of size $m$,
 which corresponds to a $(2s+1)$-independent set $B$ in $G$ that is contained in $A$; thus, we may put $S=\emptyset$ again.
 Hence, the obtained bound is $L(m)={m+t-2\choose t-1}$, and we have already argued that then $L(m)\in \Oof_{r,t}{(m^{{(4t+1)}^{2t}})}$.

 \paragraph{Even case.}
 Finally,
 we prove~\cref{lem:engine} in the case $r=2s+2$, for some integer $s\geq 1$. We put $d=9s+4=9r/2-5$.
Let $G$  be such that 
 $K_t\not\minor_{d} G$, and
let $A$ be a $(2s+1)$-independent set in~$G$. Consider the graph $G'$ obtained from $G$
by contracting the (pairwise disjoint) balls of radius $s$ around each vertex $v\in A$.
 Let $A'$ denote the set of vertices of $G'$ corresponding to the contracted balls. Again, there is a natural correspondence (bijection) between $A$ and $A'$. Note that
this time, $A'$ is a $1$-independent set in $G'$.
Since $G'\minor_s G$, from $K_t\not\minor_{9s+4} G$ it follows by \Cref{lem:combineminors} that $K_t\not\minor_4 G'$. Apply the already proved case $r=2$ to $G'$ and $A'$. 
Then, provided $|A|=|A'|\ge L_t(m)$, where $L_t(m)$ is the function as defined in the case $r=2$, we infer that
 $A'$ contains a subset $B'$ of size $m$
which is  $2$-independent in $G'-S'$, for some $S'\subset V(G')-A'$ of size less than $t$.
Since $S'\cap A'=\emptyset$, each vertex of $S'$ originates from a single vertex of $G$ before the contractions yielding $G'$; thus, $S'$ corresponds to a
set $S$ consisting of less than $t$ vertices of~$G$ which are at distance at least $s+1$ from each vertex in $A$.
In turn, the set $B'$ corresponds to some subset $B$ of $A$
which is $(2s+2)$-independent in $G-S$. Moreover, as in $G'$ each vertex of~$S'$
is a neighbor of each vertex of $B'$,  each vertex of $S$
has distance exactly $s+1=r/2$ from each vertex of $B$.

\medskip
An algorithm computing the sets $B$ and $S$ (in either the odd or even case) can be given as follows:
simply run a breadth-first search from each vertex of $A$ to compute the graph $G'$ with the balls of radius  $\lfloor \frac{r-1}2 \rfloor$  around the vertices in $A$ contracted to single vertices, 
and then run the algorithm for the case $r=1$ or $r=2$.
This yields a running time of  $\Oof_{r,t}(|A|\cdot |E(G)|)$.
 \medskip
  
This finishes  the proof of~\cref{lem:engine}.

\subsection{Proof of \cref{thm:new-uqw}}
We now wrap up the proof of \cref{thm:new-uqw} by iteratively applying~\cref{lem:engine}. 
We repeat the statement for convenience.

 \setcounter{aux}{\thetheorem}
 \setcounter{theorem}{\theuqw}
 \begin{theorem}
 For all $r,t\in \N$ there is a polynomial  $N\colon \N\to \N$ with $N(m)=
 \Oof_{r,t}{(m^{{(4t+1)}^{2rt}})}$, such that the following holds.
 Let $G$ be a graph such that $K_t\not\minor_{\lfloor 9r/2\rfloor} G$, and
 let $A\subseteq V(G)$ be a vertex subset of size at least $N(m)$, for a given $m$.
 Then there exists a set $S\subseteq V(G)$ of size $|S|<t$ and a set $B\subseteq A\setminus S$
 of size $|B|\geq m$ which is $r$-independent in $G-S$.
 Moreover, given~$G$ and $A$, such sets $S$ and $B$ can be computed in time $\Oof_{r,t}(|A|\cdot |E(G)|)$.
 \end{theorem}
 \setcounter{theorem}{\theaux}
\begin{proof}
Fix integers $r,t$,  and a graph $G$ such that $K_t\not\minor_{d} G$,
for $d=\lfloor 9r/2 \rfloor$. Let $\beta>4t$ be a fixed real. As in the proof of \cref{lem:engine}, we suppose $m\geq t^8$; this will be taken care by the final choice of the function $N(m)$.
Denote $\gamma=\beta^{2t}$, and
define the function $L(m)$ as $L(m)=((t+1)\cdot m)^\gamma$.

Define sequence $m_0,m_1,\ldots,m_r$ as follows:
\begin{eqnarray*}
m_r & = & m\\
m_i & = & L(m_{i+1}) \qquad \textrm{for }0\leq i<m.
\end{eqnarray*}
A straightforward induction yields that 
\begin{equation*}
m_i=(t+1)^{\frac{\gamma^{r-i}-1}{\gamma-1}}\cdot m^{\gamma^{r-i}}\qquad \textrm{for all }i\in \set{0,\ldots,r}.
\end{equation*}

Suppose that $A$ is a set of vertices of $G$ such that $|A|\ge m_0=(t+1)^{\frac{\gamma^{r}-1}{\gamma-1}}\cdot m^{\gamma^{r}}$. 
We inductively construct sequences of sets $A= A_0\supseteq A_1\supseteq \ldots \supseteq A_r$ and $\emptyset=S_0\subseteq S_1\subseteq S_2\ldots$
satisfying the following conditions:
\begin{itemize}
	\item $|A_i|\ge m_i=L(m_{i+1})$,
	\item $A_i\cap S_i=\emptyset$ and $A_i$ is $i$-independent in $G-S_i$.
\end{itemize}
To construct $A_{i+1}$ out of $A_i$, apply~\cref{lem:engine} to the graph $G-S_i$ and 
the $i$-independent set $A_i$ of size at least $L(m_{i+1})$. This yields a set $S\subseteq V(G)$ which is disjoint from $S_i\cup A_i$, and a subset $A_{i+1}$ of $A_i-S$ of size 
at least $m_{i+1}$
which is $(i+1)$-independent in $G-S_{i+1}$, where $S_{i+1}=S\cup S_i$. This completes the inductive construction.

In particular,  $|A_r|\ge m_r=m$ and $A_r$ is a subset of $A$ which is $r$-independent in $G-S_r$.
Observe that by construction, $|S_r|<r t/2$, as in the odd steps, the constructed set $S$ is empty, and in the even steps, it has less than $t$ elements. 
We show that in fact we have $|S_r|<t$ using the following argument, similar to the one used in~\cref{claim:at-most-t}.

By the last part of the statement of~\cref{lem:engine},  at the $i$th step of the construction, each vertex of the set $S$ obtained from \cref{lem:engine}
is at distance exactly $i/2$ from all the vertices in $A_{i+1}$ in the graph 
$G-S_i$. 
For $a\in A_r$, let $\overline{N}(a)$ denote the $\lfloor r/2\rfloor$-neighborhood of $a$ in $G-S_r$; note that sets $\overline{N}(a)$ are pairwise disjoint.
The above remark implies that each vertex $v$ of the final set $S_r$ has a neighbor in the set $\overline{N}(a)$ for each $a\in A_r$.
Indeed, suppose $v$ belonged to the set $S$ added to $S_r$ in the $i$th step of the construction; i.e. $v\in S_{i+1}\setminus S_i$.
Then there exists a path in $G-S_i$ from $v$ to $a$ of length exactly~$i/2$, which traverses only vertices at distance less than $i/2$ from $a$.
Since in this and further steps of the construction we were removing only vertices at distance at least $i/2$ from $a$, this path stays intact in $G-S_r$ and hence is completely contained in $\overline{N}(a)$.

By assumption that $m\ge t$, we may choose pairwise different vertices $a_1,\ldots,a_t\in A_r$.
To reach a contradiction, suppose that $S_r$ contains $t$ distinct vertices $s_1,\ldots,s_t$. 
By the above, the sets $\overline{N}(a_i)\cup\set{s_i}$ 
form a minor model of $K_t$ in $G$ at depth-$(\lfloor r/2\rfloor+1)$.
This contradicts the assumption that $K_t\not\minor_d G$ for $d=\lfloor 9r/2 \rfloor$.
Hence, $|S|<t$.

Define the function  $N:\N\to\N$
as $N(m)=(t+1)^{\frac{\gamma^{r}-1}{\gamma-1}}\cdot m^{\gamma^{r}}$
for $m\ge t^8$ and $N(m)=N(t^8)$ for $m<t^8$; this justifies the assumption $m\geq t^8$ made in the beginning.
Recalling that $\gamma=\beta^{2t}$ and putting $\beta=4t+1$, we
note that $N(m)\in \Oof_{r,t}{(m^{{(4t+1)}^{2rt}})}$.
The argument above shows that if $|A|\ge N(m)$, then 
there is a set $S\subset V(G)$, equal to $S_r$ above,
and a set $B\subset A$, equal to $A_r$ above,
so that $B$ is $r$-independent in $G-S$.
Given~$G$ and $A$, the sets~$S$ and $B$ can be computed by applying the algorithm of \cref{lem:engine} at most~$r$ times, so in time $\Oof_{r,t}(|A|\cdot |E(G)|)$.
This finishes the proof of~\cref{thm:new-uqw}.
\end{proof}

\section{Uniform quasi-widness for tuples}\label{sec:uqw-tuples}
We now formulate and prove an extension of~\cref{thm:new-uqw}
which applies to sets of tuples of vertices, rather than sets of vertices. 
This more general result will be used later on in the paper. 
The result and its proof are essentially adaptations to the finite of their infinite analogues introduced by Podewski and Ziegler (cf.~\cite{podewski1978stable},  Corollary 3),
modulo the numerical bounds.

We generalize the notion of independence to sets of tuples of vertices.
Fix a graph $G$ and a number $r\in \N$, and let $S\subseteq V(G)$ be a subset of vertices of $G$.
We say that vertices $u$ and $v$ are {\em{$r$-separated}} by $S$ in $G$ if every path of length at most $r$ connecting $u$ and $v$ in $G$ passes through a vertex of $S$.
We extend this notion to tuples:
two tuples $\bar u,\bar v$ of vertices of $G$ are \emph{$r$-separated} by $S$ every vertex appearing in $\bar u$ is $r$-separated by $S$ from every vertex appearing in $\bar{v}$.
Finally, if $A\subseteq V(G)^d$ is a set of $d$-tuples of vertices, for some $d\in\N$,
then we say that $A$ is \emph{mutually $r$-separated} by $S$ in $G$ 
if any two distinct $\bar u,\bar v\in A$ are $r$-separated by $S$ in $G$.

\newcommand{\uqw}{\mathrm{UQW}}
\newcommand{\puqw}{\mathrm{PUQW}}
With these definitions set, we may introduce the notion of uniform quasi-wideness for tuples.

\begin{definition}
Fix a class $\cal C$ and numbers $r,d\in\N$.
For a function $N\from\N\to\N$
and number $s\in\N$,
we say that $\cal C$ satisfies property
$\uqw^d_r(N,s)$ if the following condition holds:
   \begin{quote}\itshape 
      for every $m\in \N$ and every subset 
     $A\subseteq V(G)^d$ with $|A|\ge N(m)$, there is a set $S\subset V(G)$ with $|S|\le s$ and a subset $B\subset A$ with $|B|\ge m$ which is mutually $r$-separated by $S$ in $G$.
   \end{quote}   
    We say that $\cal C$ satisfies property $\uqw^d_r$ if  $\cal C$ satisfies $\uqw^d_r(N,s)$ for 
	some $N\from\N\to\N$ and $s\in\N$.
	If moreover one can take $N$ to be a polynomial,
	then we say that $\cal C$ satisfies property $\puqw^d_r$.
\end{definition}

When $d=1$, we omit it from the superscripts.
  Note that there is a slight discrepancy 
  in the definition of uniform quasi-wideness 
  and the property of satisfying $\uqw_r$, for all $r\in \N$.
  This is due to the fact that in the original definition,
  the set $B$ must be disjoint from $S$,
  whereas in the property $\uqw_r$, 
  some vertices of $S$ may belong to $B$. This distinction is inessential when it comes to dimension $1$, since $|S|\le s_r$ for some constant $s_r$,
  so passing from one definition to the other requires 
  modifying the function $N_r$ by an additive constant $s_r$.
In particular, a class of graphs $\cal C$ is uniformly quasi-wide if and only if it 
	satisfies $\uqw_r$, for all $r\in \N$.  
  However, generalizing to tuples of dimension $d$ requires the use of the definition above, where the tuples in~$B$ are allowed to contain  vertices which occur in $S$. 
  For example, if the graph $G$ is a star with many arms and $A$ consists of all pairs of adjacent vertices in $G$, then $S$
  needs to contain the central vertex of $G$,
  and therefore $S$ will contain a vertex from every tuple in $A$. We may take $B$ to be equal to $A$ in this case.
  
\medskip
	Using the above terminology, \cref{thm:new-uqw}
	states that for every fixed $r\in\N$, if there is a number $t\in\N$
	such that $K_t\not\minor_{\lfloor 9r/2\rfloor} G$ for all $G\in \cal C$,
	then $\cal C$ satisfies $\puqw_r$,
	and more precisely $\uqw_r(N_r,s_r)$
	for a polynomial $N_r\from\N\to\N$ and number $s_r\in \N$, where $N_r$ and $s_r$ can be computed from $r$ and $t$.
	The following result provides a generalization to higher dimensions.

\begin{theorem}\label{thm:uqw-tuples}If $\cal C$
	is a nowhere dense class of graphs,
	then for all $r,d\in\N$,
	the class $\cal C$ satisfies
	 $\puqw^d_r$.
	More precisely, for any class of graphs $\cal C$ and numbers $r,t\in\N$,
	if  	$K_t\not\minor_{18r} G$ for all $G\in \cal C$,
then for all $d\in \N$ the class $\cal C$ satisfies $\uqw^d_{r}(N^d_r,s^d_r,)$	for 
some number $s^d_r\in \N$ and polynomial $N^d_r\from \N\to\N$ that can be computed given $r$, $t$, and $d$.
\end{theorem}

\cref{thm:uqw-tuples} is an immediate consequence  of~\cref{thm:new-uqw} (or~\cref{thm:krs} if only the first part of the statement is concerned)
and of the following result.

\begin{proposition}\label{prop:uqw-tuples}
For all $r,d\in\N$,
if $\cal C$ satisfies $\uqw_{2r}(N_{2r},s_{2r})$ 
for some $s_{2r}\in\N$ and \mbox{$N_{2r}\colon\N\to \N$},
then $\cal C$
satisfies $\uqw^d_r(N^d_r,s^d_r)$
for  $s^d_r=d\cdot s_{2r}$ and 
function $N^d_r\colon \N\to \N$ defined as $N^d_r(m)=f^d((d^2+1)\cdot m)$, where $f(m')=m'\cdot N_{2r}(m')$ and $f^d$ is the $d$-fold composition of $f$ with itself.
\end{proposition}

The rest of~\cref{sec:uqw-tuples} is devoted to the proof of \cref{prop:uqw-tuples}.
Fix a class $\cal C$
such that $\uqw_{2r}(N_{2r},s_{2r})$ holds for some number $s_{2r}\in \N$ and  function $N_{2r}\from \N\to \N$.
We also fix the function $f$ defined in the statement of \cref{prop:uqw-tuples}.

%In the proof below, we will invoke property $\uqw_{2r}$ for $\cal C$,
%but also property $\uqw_r$, which also holds for $\cal C$, as witnessed e.g. by the function $N_r(\cdot)\coloneqq N_{2r}(\cdot)$ and constant $s_r\coloneqq s_{2r}$.

\medskip

	Let us fix dimension $d\in \N$, radius $r\in \N$, and graph $G\in \cal C$.
        For a coordinate $i\in\set{1,\ldots,d}$, by $\pi_i\colon V(G)^d\to V(G)$ we denote the {\em{projection}} onto the $i$th coordinate; that is,
        for $\bar{x}\in V(G)^d$ by $\pi_i(\bar{x})$ we denote the $i$th coordinate of $\bar{x}$.
        
        %and set of tuples $A\subseteq V(G)^d$, by
        %$\pi_i(A)$ we denote the {\em{multiset}} of vertices appearing on the $i$th coordinate of the tuples in $A$.
        %That is, each tuple $\bar u\in A$ contributes with one element to $\pi_i(A)$, this element being the vertex on the $i$th coordinate of $\bar u$.
        %The notion of mutual $r$-separation is naturally extended to multisets: a multiset $M$ of vertices of $G$ is {\em{mutually $r$-separated}} by a vertex subset $S\subseteq V(G)$ in $G$ if for
        %any two distinct vertices $u,v$ drawn from $M$, it holds that $u$ and $v$ are $r$-separated by $S$ in $G$.
        %Note that in case $M$ contains more than one copy of some vertex $u$, any set $S$ that mutually $r$-separates $S$ has to contain $u$.
        %By a slight abuse of notation, for a single tuple $\bar{u}$ by $\pi_i(\bar{u})$ we denote the $i$th coordinate of $\bar{u}$.
	
	Our first goal is to find a large subset of tuples that are mutually $2r$-separated by some small~$S$ on each coordinate separately.
	Note that in the following statement we ask for $2r$-separation, instead of $r$-separation.

\begin{lemma}\label{lem:step1} For all $r,m\in \N$ and $A\subset V(G)^d$ with $|A|\ge f^d(m)$,
	there is a set $B\subset A$ with $|B|\ge m$ and a set $S\subset V(G)$ with $|S|\le d\cdot s_{2r}$ 
	such that for each coordinate $i\in\set{1,\ldots,d}$ and all distinct $\bar x,\bar y\in B$,
        the vertices $\pi_i(\bar x)$ and $\pi_i(\bar y)$ are $2r$-separated by $S$. 
\end{lemma}
\begin{proof}
We will iteratively apply the following claim.

%Let $f\colon \N\to \N$ be defined as $f(m)=N(r,m)\cdot m$ for $m\in\N$.

\begin{claim}\label{claim:ith-coord}
Fix a coordinate $i\in\set{1,\ldots,d}$, an integer $m'\in\N$, and a  set $A'\subset V(G)^d$ with  $|A'|\ge f(m')$.
Then there is a set $B'\subset A'$ with $|B'|\ge m'$
and a set $S'\subset V(G)$ with $|S'|\le  s_{2r}$, such that for all distinct $\bar x,\bar y\in B$,
the vertices $\pi_i(\bar x)$ and $\pi_i(\bar y)$ are $2r$-separated by $S$.
\end{claim}
\begin{clproof}
We consider two cases, depending on whether $|\pi_i(A')|\geq N_{2r}(m')$.

Suppose first that $\pi_i(A')$ contains at least $N_{2r}(m')$ distinct vertices.
Then we may apply the property $\uqw_{2r}$ to $\pi_i(A')$, yielding sets $S'\subset V(G)$ and $X\subseteq \pi_i(A')$
such that $|X|\ge m'$, $|S'|\le s_{2r}$, and $X$ is mutually $2r$-separated by $S'$ in $G$. 
Let $B'\subseteq A'$ be a subset of tuples constructed as follows: for each $u\in X$, include in $B'$ one arbitrarily chosen tuple $\bar x\in A'$ such that the $i$th coordinate of $\bar x$ is $u$.
Clearly $|B'|=|X|\ge m'$ and for all distinct $\bar x,\bar y\in B'$, we have that $\pi_i(\bar x)$ and $\pi_i(\bar y)$ are different and $2r$-separated by $S'$ in $G$; this is because $X$ is mutually $2r$-separated by $S'$
in $G$. Hence $B'$ and $S'$ satisfy all the required properties.

Suppose now that $|\pi_i(A')|<N_{2r}(m')$. 
Then choose a vertex $a\in \pi_i(A')$ for which the pre-image $\pi_i^{-1}(a)$ has the largest cardinality.
Since $|A'|\geq f(m')=m'\cdot N_{2r}(m')$, we have that 
$$|\pi_i^{-1}(a)|\geq \frac{|A'|}{|\pi_i(A')|}\geq \frac{m'\cdot N_{2r}(m')}{N_{2r}(m')}=m'.$$
Hence, provided we set $S'=\set{a}$ and $B'=\pi_i^{-1}(a)$, we have that $B'$ is mutually $2r$-separated by~$S'$, $|B'|\geq m$, and $|S'|=1$.
\end{clproof}

We proceed with the proof of \cref{lem:step1}.
Let $A\subset V(G)^d$ be such that $|A|\ge f^d(m)$.
We inductively define subsets $B_0\supseteq B_1\supseteq \ldots \supseteq B_d$ of $A$ and sets $S_1,\ldots,S_d\subseteq V(G)$ as follows.
First put $B_0=A$. Then, for each $i=1,\ldots,d$,
let $B_{i}$ and $S_i$ be the $B'$ and $S'$ obtained from \cref{claim:ith-coord} applied to the set of tuples $B_{i-1}\subset V(G)^d$, the coordinate $i$, and $m'=f^{d-i}(m)$. 
It is straightforward to see that the following invariant holds for each $i\in \set{1,\ldots,d}$: $|B_i|\ge f^{d-i}(m)$ and for all $j\leq i$
and distinct $\bar x,\bar y\in B_i$, the vertices $\pi_j(\bar x)$ and $\pi_j(\bar{y})$ are $2r$-separated by $S_1\cup\ldots\cup S_i$ in $G$.
In particular, by taking $B=B_d$ and $S=S_1\cup\ldots \cup S_d$, we obtain that $|B|\ge m$, $|S|\le d\cdot s_{2r}$, and $B$ and $S$ satisfy the condition requested in the lemma statement.
\end{proof}

The next lemma will be used to turn mutual $2r$-separation on each coordinate to mutual $r$-separation of the whole tuple set.

\begin{lemma}\label{lem:step2}
	Let $B\subset V(G)^d$ and $S\subset V(G)$ be such that 
   for each $i\in \set{1,\ldots,d}$ and all distinct $\bar{x},\bar{y}\in B$, the vertices $\pi_i(\bar{x})$ and $\pi_i(\bar{y})$ are $2r$-separated by $S$ in $G$.
	Then there is a set $C$ with $C\subset B$ and $|C|\geq\frac{|B|}{d^2+1}$
	such that $C$ is mutually $r$-separated by $S$ in $G$.
\end{lemma}
\begin{proof}
Let $C$ be a maximal subset of $B$ that is mutually $r$-separated by $S$ in $G$.
By the maximality of $C$, with each tuple $\bar a\in B-C$ we may associate a tuple $\bar b\in C$ and a pair of indices $(i,j)\in \set{1,\ldots,d}^2$ that witness that $a$ cannot be added to $C$, namely
$\pi_i(\bar a)$ and $\pi_j(\bar b)$ are not $r$-separated by $S$ in $G$.
Observe that two different tuples $\bar a,\bar a'\in B-C$ cannot be associated with exactly the same $\bar b\in C$ and same pair of indices $(i,j)$.
Indeed, then both $\pi_i(\bar a)$ and $\pi_i(\bar a')$ would not be $r$-separated from $\pi_j(\bar b)$ by $S$ in $G$, 
which would imply that $\pi_i(\bar a)$ and $\pi_i(\bar a')$ would not be $2r$-separated from each other by $S$,
a contradiction with the assumption on $B$.
Hence, $|B-C|$ is upper bounded by the number of tuples of the form $(\bar b,i,j)\in C\times \set{1,\ldots,d}^2$, which is $d^2|C|$.
We conclude that $|B-C|\leq d^2|C|$, which implies $|C|\geq \frac{|B|}{d^2+1}$.
\end{proof}

To finish the proof of \cref{prop:uqw-tuples},
given a set $A\subset V(G)^d$ and integer $m\in\N$,
first apply 
\cref{lem:step1} 
  with $m'= m\cdot (d^2+1)$.
 Assuming that $|A|\ge f^d(m')$, 
we obtain a set $B\subseteq A$ with $|B|\ge m\cdot (d^2+1)$ and a set $S\subset V(G)$ with $|S|\le d\cdot s_{2r}$,
such that for each $i\in \set{1,\ldots,d}$ and all distinct $\bar{x},\bar{y}\in B$, the vertices $\pi_i(\bar{x})$ and $\pi_i(\bar{y})$ are $2r$-separated by $S$ in $G$. 
Then, apply \cref{lem:step2} to $B$ and $S$, yielding a set $C\subset B$ which is mutually $r$-separated by $S$ and has size at least $m$. 
This concludes the proof of \cref{prop:uqw-tuples}.

\section{Types and locality}\label{sec:gaifman}
In this section, we develop auxiliary tools concerning first order logic on graphs, 
in particular we develop a convenient abstraction for Gaifman's locality property that can be easily combined with the notion of $r$-separation.
We begin by recalling some standard notions from logic.

\subsection{Logical notions}

\paragraph{Formulas.}
All formulas in this paper are first order formulas on graph,
i.e., they are built using variables (denoted $x,y,z$, etc.),
atomic predicates $x=y$ or $E(x,y)$,
where the latter denotes the existence of an edge between two nodes, quantifiers $\forall x,\exists x$, and boolean connectives $\lor,\land,\neg$. 
Let $\phi(\bar x)$ be a formula with free variables 
$\bar x$. (Formally, the free variables form a set.
To ease notation, we identify this set with a tuple by fixing any its enumeration.)
If $\bar w\in V^{|\bar x|}$ is a tuple of vertices of some graph $G=(V,E)$ (treated as a valuation of the free variables $\bar x$), then we write $G,\bar w\models \phi(\bar x)$
to denote that the valuation $\bar w$ satisfies the formula $\phi$ in the graph $G$.
The following example should clarify our notation.

\begin{example}\label{ex:dist-formula}
The formula
$$\phi(x,y)\equiv \exists z_1\, \exists z_2\, (E(x,z_1)\lor (x=z_1))\land (E(z_1,z_2)\lor (z_1=z_2))\land (E(z_2,y)\lor (z_2=y))$$
with free variables $x,t$ expresses that $x$ and $y$ are at distance at most $3$.
That is, for two vertices $u,v$ of a graph $G$,
the relation $G,u,v\models \phi(x,y)$ holds 
if and only if the distance between $u$ and $v$ is at most $3$ in $G$.
\end{example}

We will consider also \emph{colored graphs},
where we have a fixed set of colors $\Lambda$ and every vertex is assigned a subset of colors from 
$\Lambda$. If $C\in \Lambda$ is a color then the atomic formula $C(x)$ holds in a vertex $x$ if and only if $x$ has color $C$.

Finally, we will consider \emph{formulas with parameters}
from a set $A$, which is a subset of vertices of some graph.
Formally, such formula with parameters is a pair consisting of a (standard) formula $\phi(\bar x,\bar y)$
with a partitioning of its free variables into $\bar x$ and $\bar y$,
and a valuation $\bar v\in A^{|\bar y|}$ of the free variables $\bar y$ in $A$.
We denote the resulting formula with parameters by $\phi(\bar x,\bar v)$, and say that its free variables 
are $\bar x$. For a valuation $\bar u\in A^{|\bar x|}$,
we write $G,\bar u\models \phi(\bar x,\bar v)$
iff $G,\bar u\bar v\models \phi(\bar x,\bar y)$. Here and later on, we write $\bar u\bar v$ for the concatenation of tuples $\bar u$ and $\bar v$.

\newcommand{\tp}{\mathrm{tp}}

\paragraph{Types.}
Fix a formula $\phi(\bar x,\bar y)$ together with a distinguished partitioning of its free variables into 
\emph{object variables} $\bar x$ and \emph{parameter variables} $\bar y$. 
Let $G=(V,E)$ be a graph, and let $A\subset V$.
If $\bar u\in V^{|\bar y|}$ is a tuple of 
nodes of length $|\bar y|$, then the 
\emph{$\phi$-type of $\bar u$ over $A$},
denoted $\tp^\phi_G(\bar u/A)$,
is the set of all
formulas $\phi(\bar x,\bar v)$,
with parameters $\bar v\in A^{|\bar y|}$
replacing the parameter variables $\bar z$,
such that $G,\bar u\models \phi(\bar x,\bar v)$.
Note that since $\phi$ is fixed in this definition, formulas $\phi(\bar x,\bar v)$ belonging to the $\phi$-type of $\bar u$ are in one-to-one correspondence
with tuples $\bar v\in A^{|\bar y|}$ satisfying $G,\bar u,\bar v\models \phi(\bar u,\bar v)$.
Therefore, up to this bijection, we have the following identification:
\begin{equation}\label{eq:bijection}
\tp^\phi_G(\bar u/A)\quad\leftrightarrow\quad\setof{\bar v\in  A^{|\bar y|}}{G, \bar u\bar v\models \phi(\bar x,\bar y)}.
\end{equation}

If $q\in \N$ is a number and $\bar u\in  V^{d}$
is a tuple of some length $d$, then by $\tp^q_G(\bar u/A)$  we denote the set of all formulas $\phi(\bar x,\bar v)$
of quantifier rank at most $q$, with parameters $\bar v$ from $A$, and with $|\bar x|=d$,
such that $G,\bar u\models \phi(\bar y,\bar v)$.
Therefore, up to the correspondence \eqref{eq:bijection}, we have the following identification:
\begin{equation*}
\tp^q_G(\bar u/A)\quad\leftrightarrow\quad\set{\tp^\phi(\bar u/A)}_{\phi(\bar x,\bar y)},
\end{equation*}
where $\phi(\bar x,\bar y)$ ranges over all formulas of quantifier rank $q$, and all partitions of its free variables into two sets $\bar x,\bar y$,
where $|\bar x|=d$. 
In particular, the set $\tp^q_G(\bar u/A)$ is infinite.
It is not difficult to see, however, that in the case when $A$ is finite,
the set $\tp^q_G(\bar u/A)$ is uniquely determined by its finite subset, since up to syntactic equivalence, 
there are only finitely many formulas of quantifier rank $q$ with $|\bar u|$ free variables and parameters from $A$
(we can assume that each such formula has $|A|+|\bar u|$ free variables).
In particular, the set of all possible types 
$\tp^q_G(\bar u/A)$ has cardinality upper bounded by some number 
 computable from $q,|\bar u|$ and $|A|$.

When $\Delta$ is either a formula $\phi(\bar x,\bar y)$ with a distinguished partitioning of its free variables, or a number $q$,
we simply write $\tp^\Delta(\bar u/A)$ if the graph $G$
is clear from the context.
In the case $A=\emptyset$, we omit it from the notation, 
and simply write $\tp^\Delta(\bar u)$ or $\tp^\Delta_G(\bar u)$.
Observe that in particular, if $\Delta=q$ and $A=\emptyset$, then $\tp^q_G(\bar u)$ consists of all first order formulas $\phi(\bar x)$ of quantifier rank at most~$q$ and with $|\bar x|=|\bar u|$
such that $G,\bar u\models \phi(\bar x)$. This coincides with the standard notion of the first order type of quantifier rank $q$ of the tuple $\bar u$.

\begin{example}
Let $\phi(x,y)$	be the formula from~\cref{ex:dist-formula}, denoting that the distance between~$x$ and $y$ is at most $3$.
We  partition  the free variables of $\phi$
into $x$ and $y$.
Let $A$ be a subset of vertices of a graph $G=(V,E)$
and $u\in V$ be a single vertex.
The $\phi$-type of $u$ over $A$
corresponds, via the said bijection, to the set of those vertices in $A$
whose distance from $u$ is at most $3$ in $G$.
\end{example}

For a fixed formula $\phi(\bar y,\bar z)$,  graph $G=(V,E)$ and sets $A,W\subset V$, define
 $S^\phi(W/A)$ as the set of all $\phi$-types of tuples from $W$ over $A$ in $G$; that is, 
\begin{equation*}
S^\phi(W/A)=\setof{\tp^\phi_G(\bar u/A)}{\bar u\in W^{|\bar y|}}.
\end{equation*}
Although not visible in the notation, the set $S^\phi(W/A)$ depends on the chosen partitioning $\bar x,\bar y$ of the free variables of $\phi$.
In case $W=V(G)$ we write $S^{\phi}_d(G/A)$ instead of $S^{\phi}_d(W/A)$.
Note that this definitions differs syntactically from the one given in \cref{sec:intro}, as here $S^{\phi}(G/A)$ consists of $\phi$-types, and not of subsets of tuples.
However, as we argued, there is a one-to-one correspondence between them, as expressed in~\eqref{eq:bijection}.

The following lemma is immediate.
\begin{lemma}\label{lem:types-over-B}
Let $G$ be a graph and let $A\subseteq B\subseteq V(G)$. Then for each formula $\phi(\bar x,\bar y)$, it holds that
$|S^\phi(G/A)|\leq |S^\phi(G/B)|$. 
\end{lemma}

\subsection{Locality}
We will use the following intuitive notion of functional determination.
Suppose $X,A,B$ are sets and we have two functions: $f\colon X\to A$ and $g\colon X\to B$.
We say that $f(x)$ {\em{determines}} $g(x)$ for $x\in X$ if for every pair of elements $x,x'\in X$ the following implication holds: $f(x)=f(x')$ implies $g(x)=g(x')$.
Equivalently, there is a function $h\colon A\to B$ such that $g=h\circ f$.
%Note that this relation is transitive: if $f(x)$ determines $g(x)$ and $g(x)$ determines $h(x)$, then $f(x)$ determines $h(x)$.

Recall that if $A,B,S$ are subsets of vertices of a graph $G$ and $r\in\N$,
then $A$ and $B$ are $r$-separated by $S$ in $G$
if every path from $A$ to $B$ of length at most $r$ contains a vertex from $S$.

\medskip
The following lemma is the main result of this section. 

\begin{lemma}\label{lem:types}
For any given numbers $q$ and $d$
one can compute numbers $p$ and $r$ with the following properties.
Let $G=(V,E)$ be a fixed graph and let $A,B,S\subset V$ be fixed subsets of its vertices
such that $A$ and $B$ are $r$-separated by~$S$ in $G$.
Then, for tuples $\bar u\in A^{d}$, the type $\tp^q(\bar u/B)$ is determined by the type $\tp^{p}(\bar u/S)$.
\end{lemma}

We will only use the following consequence of~\cref{lem:types}.

\begin{corollary}\label{cor:bound}
For every formula $\phi(\bar x,\bar y)$ 
and number $s\in \N$
there exist numbers $T,r\in \N$,
where~$r$ is computable from $\phi$ and $T$ is computable from $\phi$ and $s$,
  such that the following holds. For every graph $G$ and vertex subsets $A,B,S\subset V(G)$ 
  where~$S$ has at most $s$ vertices and $r$-separates $A$ from $B$, we have $|S^\phi(A/B)|\le T$.
\end{corollary}
\begin{proof}
Apply~\cref{lem:types} to $q$ being the quantifier rank of $\phi$ and $d=|\bar x|$, yielding numbers $p$ and $r$.
By~\cref{lem:types} we have $|S^\phi(A/B)|\leq |S^\phi(A/S)|$.
However, $|S^\phi(A/S)|$ is the number of quantifier rank $p$ types of $d$-tuples of elements over a set of parameters of size $s$, and, as we argued, this number is bounded by a value computable from $p$, $d$, and $s$.
\end{proof}

The remainder of this section is devoted to the proof of~\cref{lem:types}.
This result is a consequence of two fundamental properties of first order logic:
Gaifman's locality and Feferman-Vaught compositionality. We recall these results now.
The following statement is an immediate corollary of the main result in a paper of Gaifman~\cite{gaifman1982local}.

\begin{lemma}[Gaifman locality,~\cite{gaifman1982local}]\label{lem:gaif}
  For all numbers $d,q\in \N$ there exists a number $t\in \N$, computable from $d$ and $q$, such that the following holds.
  Let $G=(V,E)$ be a graph colored by a fixed set colors, and $A\subset V$ be a set of vertices of $G$.
  Then, for tuples $\bar u\in V^d$, the type  $\tp^q(\bar u)$ is determined by the type $\tp^{t}(B^r(\bar u))$, where $r=7^q$.
\end{lemma}

The next result expresses compositionality of first order logic. Its proof is a standard application of Ehrenfeucht-Fra\"iss\'e games, so we only sketch it for completeness.

\begin{lemma}[Feferman-Vaught]\label{lem:fv}
  Let $G,H$ be two fixed vertex-disjoint graphs colored by a fixed set of colors $\Lambda$, and let 
  $c,d\in\N$  be numbers.
  Then, for valuations $\bar u\in V(G)^{c}$ and $\bar v\in V(H)^{d}$, 
 the type 
 $\tp^q_{G\cup H}(\bar u\bar v)$
 is determined by the pair of types $\tp^q_G(\bar u)$ and $\tp^q_H(\bar v)$.
\end{lemma}
\begin{proof}[sketch]The proof proceeds by applying the following, well-known characterization of $\tp^q_G(\bar w)$ in terms of Ehrenfeucht-Fra\"iss\'e games:
$\tp^q_{G}(\bar w)=\tp^q_{G}(\bar w')$
if and only if duplicator has a strategy to survive for $q$-rounds in a certain pebble game.
To prove the lemma, we combine two strategies of duplicator: one on $G$ and one on $H$.
\end{proof}

Before proving~\cref{lem:types}, we introduce the following notions.
Fix a graph $G=(V,E)$.
For a set of vertices $S\subset V$, define the color set $\Lambda_S=\{C_s\colon s\in S\}$, where we put one color $C_s$ for each vertex $s\in S$.
Define a graph $G^S$ colored with $\Lambda_S$, which is 
the subgraph of $G$ induced by $V-S$
in which, additionally, for every vertex $s\in S$, every vertex $v\in V-S$ which is a neighbor of $s$ in $G$ is colored by color $C_s$. 
In other words, every vertex $v$ of $G^S$ is colored with a subset of colors from $\Lambda_S$ corresponding to the neighborhood of $v$ in $S$.

A sequence of elements of $S\cup\set\star$,
where $\star$ is a fixed placeholder symbol,
will be called an \emph{$S$-signature}.
If $H$ is any (colored) graph with vertex set contained in $V-S$,
and $\bar u\in V^d$ is a $d$-tuple of vertices,
define the {$S$-signature} of $\bar u$
as the tuple $\bar s\in (S\cup\set\star)^d$ obtained from $\bar u$ by replacing the vertices in $V-S$ by the symbol~$\star$.
Define $\tp^q[H,\bar u]$ as the
pair consisting of the following components:
\begin{itemize}
	\item the type $\tp^q_H(\bar v)$,
	where $\bar v$ is the tuple obtained from $\bar u$
	by removing the vertices which belong to $S$; and
	\item the $S$-signature of $\bar u$.
\end{itemize}

Given a graph $G$, a subset of its vertices $S$, and a tuple of vertices $\bar u$,
by $N^r_S(\bar u)$ we denote the subgraph of $G^S$ induced by the set of vertices reachable from a vertex in $\bar u$ by a path of length at most $r$
in $G-S$ (the graph $G$ is implicit in the notation $N^r_S(\bar u)$). Note that $N^r_S(\bar u)$ is a colored graph, with colors inherited from $G^S$.

\begin{comment}
We will prove the following strengthening of~\cref{lem:types}:

\begin{lemma}%[Gaifman locality $\lor$ Feferman-Vaught]
  \label{lem:types1}
For any given number $q\in\N$ one can compute 
 a number $r\in\N$ with the following property.
	For any graph $G=(V,E)$, sets of vertices $A,B,S\subset V$	
	such that $A$  and $B$ are $r$-separated by $S$,
	for every tuple $\bar u\in A^{d}$, 
	the type $\tp^q_G(\bar u/B)$
	is computable from $\tp^{q}[B^r_S(\bar u), \bar u]$, and $G$ and $S$.

		%
	% and any $\bar u,\bar v\in B^{\bar y}$, the following implication holds:
	% $$\text{if\quad}\tp^q(\bar u/S)=\tp^q(\bar v/S)\text{\quad then\quad}
	% 	\tp^q(\bar u/A)=\tp^q(\bar v/A).$$
\end{lemma}

To show that~\cref{lem:types1} implies~\cref{lem:types}, define $p$ as $q\cdot r$. It suffices to show that
$\tp^{q}[B^r_S(\bar u), \bar u]$ is computable from $\tp_G^p(\bar u/S)$, and $G$ and $S$. This is the case, since
a formula $\phi(\bar y)$
can be relativized to $B^r_S(\bar u)$
by replacing each quantifier $\exists x$ by a formula
$\exists x\exists x_1\ldots\exists x_r\psi (x,x_1,\ldots,x_r)$,
where $\psi$ specifies that $x_1,\ldots,x_r,x$ form a path
starting in one of the vertices in $\bar u$, ending in $x$,
and omitting all vertices in $S$, which are enumerated as parameters.

It remains to prove~\cref{lem:types1}.
% We will use slightly stronger variants
% of~\cref{lem:gaif}
% and~\cref{lem:fv}, where the  graphs are colored with a fixed number of colors; the  types computed in the statements then depend also on the number of colors.
\end{comment}

With all these definitions and results in place, we may proceed to the proof of \cref{lem:types}.

\begin{proof}[of~\cref{lem:types}]
We prove the lemma for $r=7^q$.
Let $t$ be the constant given by Gaifman's lemma, \cref{lem:gaif}, for $q$ and $d$, and let $p=t+r$.

Fix $G,A,B,S$ as in the statement of the lemma, and fix a tuple $\bar w\in B^\ell$, for some length $\ell$.
To prove the lemma, it is enough to show that
for tuples $\bar u\in A^d$,
the type $\tp^q_G(\bar u\bar w)$ is determined by the type $\tp^p(\bar u/S)$.
Indeed, applying this to every tuple $\bar w$ of parameters from $B$ implies that $\tp^q_G(\bar u/B)$ is determined by $\tp^q_G(\bar u/S)$, as requested.

We will prove the following sequence of determinations,
where an arrow $a\rightarrow b$ signifies that $b$ is determined by $a$:
\begin{align*}
	\tp^p_G(\bar u/S)
  \ \longrightarrow\ 
	\tp^{t}[N^r_S(\bar u), \bar u]
  \ \stackrel{\textrm{(\ref{lem:fv})}}{\longrightarrow}\ 
	\tp^{t}[N^r_S(\bar u\bar w), \bar u\bar w] \ \stackrel{\textrm{(\ref{lem:gaif})}}\longrightarrow\ 
	\tp^q[G^S, \bar u\bar w] \ \longrightarrow\ 
	\tp^q_G(\bar u\bar w).
\end{align*}
The second arrow follows from Feferman-Vaught's lemma, \cref{lem:fv},
as the colored graph $N^r_S(\bar u\bar w)$
is the disjoint union of the colored graphs 
$N^r_S(\bar u)$ and $N^r_S(\bar w)$,
because $\bar u$ and $\bar w$ are $r$-separated by $S$.
The third arrow is directly implied by the Gaifman's lemma,~\cref{lem:gaif}.
We are left with arguing the first and the last arrow, which both follow from simple rewriting arguments, presented below.

\medskip
For the first arrow, obviously already $\tp^0_G(\bar u/S)$ determines the $S$-signature of $\bar u$. Let $\bar s$ be any enumeration of $S$.
To see that $\tp^p_G(\bar u/S)$ determines $\tp^{t}_{N^r_S(\bar v)}(\bar v)$, where $\bar v$ is $\bar u$ with vertices of~$S$ removed, take any formula $\phi(\bar x)$ with $|\bar x|=|\bar v|$.
Let $\phi'(\bar x,\bar s)$ be the formula with parameters~$\bar s$ from $S$ that is syntactically derived from $\phi(\bar x)$ as follows: 
to every quantification $\exists y$ in $\phi(\bar x)$ we add a guard $\delta(y,\bar s)$ stating that there is a path from some element of $\bar x$ to $y$ that has length at most $r$ and does not pass through
any vertex of $\bar s$; it is easy to see that there is such a guard $\delta(y,\bar s)$ with quantifier rank $r$.
Then $\phi'(\bar x,\bar s)$ has quantifier rank at most $t+r=p$, and it is straightforward to see that for every $\bar v\in (A-S)^{|\bar v|}$, we have $G,\bar v\models \phi'(\bar x,\bar s)$ if and only if
$N^r_S(\bar v),\bar v\models \phi(\bar x)$. Therefore, to check whether $\phi(\bar x)$ belongs to $\tp^{t}_{N^r_S(\bar v)}(\bar v)$ it suffices to check 
whether $\phi'(\bar x,\bar s)$ belongs to $\tp^p_G(\bar u/S)$, so the latter type determines the former.

\medskip

The argument for the last arrow is provided by the following claim.

\begin{claim}\label{cl:rewrite}
  Let $\phi$ be a formula
  with $k$ free variables and quantifier rank at most $q$, 
  and let $\sigma$ be an $S$-signature of length $k$.
  One can compute a formula $\phi^S$ of quantifier rank at most $q$
  whose free variables correspond to the $\star$'s in $\sigma$,
  such that for every tuple $\bar v$ of elements of $G$
  whose $S$-signature is~$\sigma$,
   $\phi(\bar v)$ holds in $G$
  if and only if $\phi^S(\bar v^S)$ holds in $G^S$, where $\bar v^S$ is obtained from $\bar v$ by removing 
  those elements that belong to $S$.
\end{claim}
\begin{clproof}[Sketch]
The proof proceeds by a straightforward induction on the structure of the formula $\phi$.
In essence, every quantification $\exists y$ of a vertex~$y$ in $G$ is replaced by quantification of $y$ in $G-S$ plus a disjunction over $s\in S$ of formulas where we assume $y=s$.
Atomic formulas of the form $E(x,y)$ and $x=y$ have to be replaced accordingly. Say for $E(x,y)$: if both~$x$ and~$y$ are assumed to be in $G-S$, then we leave $E(x,y)$ intact;
if $x$ is assumed to be in $S$ (say we assume $x=s$) and $y$ is assumed to be in $G-S$, then we substitute $E(x,y)$ by $C_s(y)$; and if both~$x$ and~$y$ are assumed to be in $S$, then 
we replace $E(x,y)$ by $\bot$ or $\top$ depending on whether the vertices assumed to be equal to $x$ and $y$ are adjacent or not.
We leave the details to the reader.
\end{clproof}

\cref{cl:rewrite}, applied to $k=|\bar u\bar w|$, implies that 
$\tp^q_G(\bar u\bar w)$ is determined by $\tp^q[G^S, \bar u\bar w]$, finishing the proof of~\cref{lem:types}.	
\end{proof}

We remark that in all the results of this section, whenever some type determines some other type, it is actually true that the latter type can be {\em{computed}} given the former type together with the graph $G$ 
and, if applicable, also the set of vertices $S$. For Gaifman's locality lemma, the effectiveness follows from the original proof of Gaifman~\cite{gaifman1982local}, and it is not hard to see that the proof of the Feferman-Vaught lemma (\cref{lem:fv}) can be also made effective. By examining our proofs, one can readily verify that all the stated determination relations can be made effective in this sense.

\section{Bounds on the number of types}\label{sec:types}

In this section we prove \Cref{thm:vc-density} and \Cref{thm:vc-density-lower-bound}.
	Recall that \Cref{thm:vc-density} provides upper bounds on the number of types in classes of graphs which are nowhere dense, 
	and stronger bounds for classes which have bounded expansion. 
	On the other hand, the complementary \Cref{thm:vc-density-lower-bound} shows that for subgraph-closed classes, in the absence of structural sparsity we cannot hope for such upper bounds.

\subsection{Upper bounds for sparse classes}
We first prove \Cref{thm:vc-density}, which we recall for convenience.

%\vcupper*
 \setcounter{aux}{\thetheorem}
 \setcounter{theorem}{\thevcupper}
 \begin{theorem}\label{thm:vc-density-recall}
 Let $\CCC$ be a class of graphs and let $\phi(\tup x,\tup y)$ be a first order formula
 with free variables  partitioned  into object variables $\bar x$  and parameter variables $\bar y$. Let $\ell=|\bar x|$. Then:
 \begin{enumerate}[(1)]
 \item If $\CCC$ is nowhere dense, then for every $\epsilon>0$
 there exists a constant~$c$ such that for every $G\in \CCC$ and every nonempty
 $A\subseteq V(G)$, we have $|S^\phi(G/A)|\leq c\cdot |A|^{\ell+\epsilon}.$
 \item If $\CCC$ has bounded expansion, then there exists a constant~$c$ such that for every $G\in \CCC$ and every nonempty $A\subseteq V(G)$, we have $|S^\phi(G/A)|\leq c\cdot |A|^\ell$.
 \end{enumerate}
 \end{theorem}
 \setcounter{theorem}{\theaux}
We remark that the theorem holds also for colored graphs, in the following sense.
A class of graphs  whose vertices and edges are colored by a fixed finite number of colors is nowhere dense if the underlying class of graphs obtained by forgetting the colors is nowhere dense. 
Then~\cref{thm:vc-density} holds also for such classes of colored graphs, with the same proof. Namely, all graph theoretic notions
are applied to the underlying colorless graphs, only the definition of types takes the colors into account. 
The results of~\cref{sec:gaifman} then need to be lifted to edge- and vertex-colored graphs, but this is straightforward.

\medskip
The proof of~\cref{thm:vc-density} spans the remainder of this section.
In this proof, we will
will first enlarge the set $A$ to a set $B$, called
an \emph{$r$-closure of~$A$} (where $r$ is chosen depending on $\phi$), such 
that the connections of elements from $V(G)-B$ 
toward $B$ are well controlled. This approach
was first used in Drange et al.~\cite{drange2016kernelization} in the context of classes of bounded expansion, 
and then for nowhere dense classes in Eickmeyer et al.~\cite{eickmeyer2016neighborhood}. 
We start by recalling these notions.

Let $G$ be a graph and let $B\subseteq V(G)$ be a subset of vertices. For vertices $v\in B$ and $u\in V(G)$, a path $P$ leading from $u$ to $v$ is called {\em{$B$-avoiding}}
if all its vertices apart from~$v$ do not belong to~$B$. Note that if $u\in B$, then there is only one $B$-avoiding path leading from $u$, namely the one-vertex path where $u=v$.
For a positive integer $r$ and $u\in V(G)$, the {\em{$r$-projection}} of $u$ on $B$, denoted $M^G_r(u,B)$, is the set of all vertices $v\in B$ such that there is
a $B$-avoiding path of length at most $r$ leading from $u$ to $v$. Note that for $u\in B$, we have $M^G_r(u,B)=\{b\}$.
Equivalently, $M^G_r(u,B)$ is the unique inclusion-minimal
subset of $B$ which $r$-separates $u$ from $B$ (cf.~Fig.~\ref{fig:projection}).

\begin{figure}[h!]
	\centering
		\includegraphics[scale=0.35,page=2]{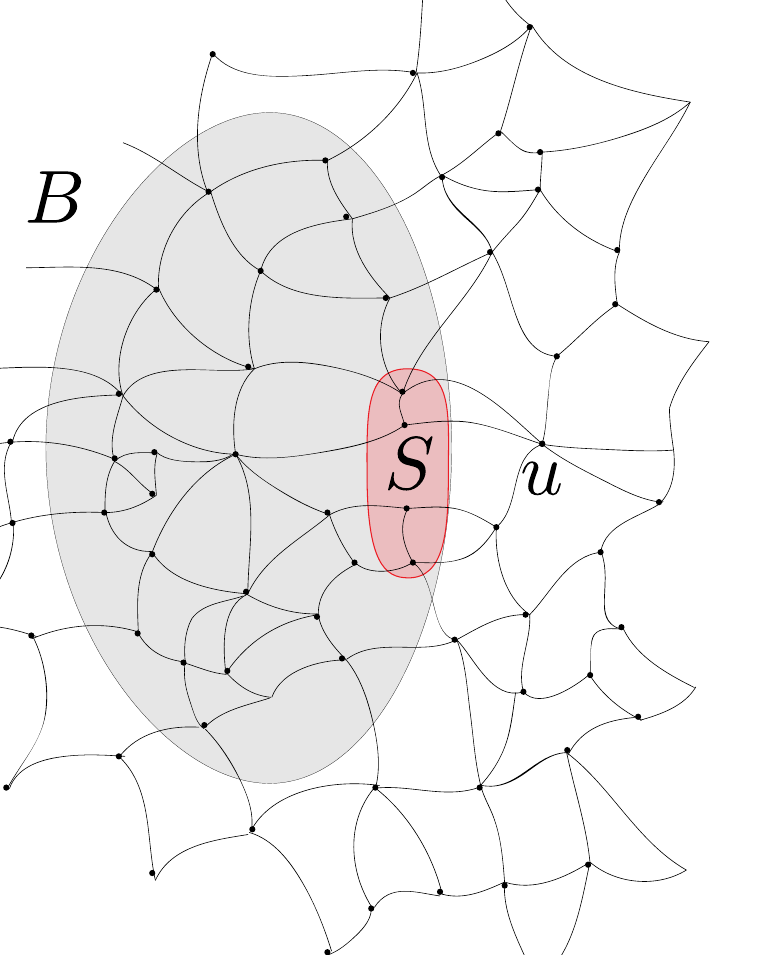}
	\caption{The  $r$-projection of $u$ on $B$
	(here $r=2$)
	is the minimal set  $S\subset B$
	which $r$-separates $ u$ from $B$.
	}
	\label{fig:projection}
\end{figure}

We will use the following results from~\cite{drange2016kernelization,eickmeyer2016neighborhood}.

\begin{lemma}[\cite{drange2016kernelization}]\label{lem:closure-be}
Let $\CCC$ be a class of bounded expansion. 
Then for every $r\in \N$ there is a constant $c\in \N$ such that for
every $G\in \CCC$ and $A\subseteq V(G)$ there exists a set $B$, called an {\em{$r$-closure}} of $A$, with the following properties:
\begin{enumerate}[(a)]
  \item $A\subseteq B\subseteq V(G)$;
  \item $|B|\leq c\cdot |A|$; and
  \item $|M_r^G(u,B)|\leq c$ for each $u\in V(G)$.
  \end{enumerate}
  Moreover, for every set $X\subset V(G)$, it holds that
  \begin{enumerate}[(d)]
  \item $|\setof{M_r^G(u,X)}{u\in V(G)}|\leq c\cdot |X|$.
\end{enumerate}
\end{lemma}

\begin{lemma}[\cite{eickmeyer2016neighborhood}]\label{lem:closure-nd}
Let $\CCC$ be a nowhere dense class. 
Then for every $r\in\N$ and $\delta>0$ there is a 
constant $c\in\N$ such that for every $G\in \CCC$ and $A\subseteq V(G)$ there exists a set 
$B$,  called an {\em{$r$-closure}} of $A$, 
with the following properties: 
\begin{enumerate}[(a)]
  \item $A\subseteq B\subseteq V(G)$;
  \item $|B|\leq c\cdot |A|^{1+\delta}$; and
  \item $|M_r^G(u,B)|\leq c\cdot |A|^{\delta}$ for each $u\in V(G)$.
  \end{enumerate}
  Moreover, for every set $X\subset V(G)$, it holds that
  \begin{enumerate}[(d)]  
  \item $|\setof{M_r^G(u,X)}{u\in V(G)}|\leq c\cdot |X|^{1+\delta}$.
\end{enumerate}
\end{lemma}

We note that in~\cite{drange2016kernelization,eickmeyer2016neighborhood} projections on $B$ are defined only for vertices outside of $B$. 
However, adding singleton projections for vertices of $B$ to the definition only adds $|B|$ possible projections of size $1$ each, so this does not influence the validity of the above results.

We proceed with the proof of \cref{thm:vc-density}.
 To focus attention, we present the proof only for the nowhere dense case (first statement). The proof in the bounded expansion case (second statement)
 can be obtained by replacing all the parameters $\epsilon,\delta,\epsilon_1,\epsilon_2$ below by~$0$, and substituting the usage of \cref{lem:closure-nd} with \cref{lem:closure-be}.

 Let us fix: a nowhere dense class of graphs $\CCC$, a graph $G\in \CCC$, a vertex subset $A\subseteq V(G)$, a real $\epsilon>0$, and 
 a first order formula $\phi(\bar x,\bar y)$, where $\bar x$ is the distinguished $\ell$-tuple of object variables.
 Our goal is to show that $|S^\phi(G/A)|=\Oof(|A|^{\ell+\epsilon})$.
 	   
In the sequel, $d$ denotes a positive integer 
depending on ${\cal C},\ell,\phi$ only (and not on $G, A$ and $\epsilon$), and will be specified later. We may choose positive reals
$\delta,\epsilon_1$ such that 
	 $(\ell+\epsilon_1)(1+\delta) 
	 \le
	 \ell+\epsilon$ and $\epsilon_1>\delta(d+\ell)> \delta\ell$, for instance as follows: $\epsilon_1=\epsilon/2$ and $\delta=\frac{\epsilon}{4d+4\ell}$.
The constants hidden in the $\Oof(\cdot)$ notation below depend
 on $\epsilon,\delta,\epsilon_1,\cal C, \ell$ and $\phi$, but not on $G$ and $A$.   By \emph{tuples} below we refer to tuples of length $\ell$.

	Let $q$ be the quantifier rank of $\phi$ and let 
$p,r$ be the numbers obtained by applying \cref{lem:types} to $q$ and $\ell$.
Let $B$ be an $r$-closure of $A$, given by~\cref{lem:closure-nd}.
  By~\cref{lem:closure-nd}, the total number of distinct $r$-projections onto $B$ 
  is at most $\Oof(|B|^{1+\delta})$, and each of these projections has size $\Oof(|B|^{\delta})$.
  	   Figure~\ref{fig:sketch} serves as  an illustration to the steps of the proof in the case $\ell=1$.
  	   \begin{figure}[h!]
  	   	\centering
  	   		\includegraphics[scale=0.346,page=4]{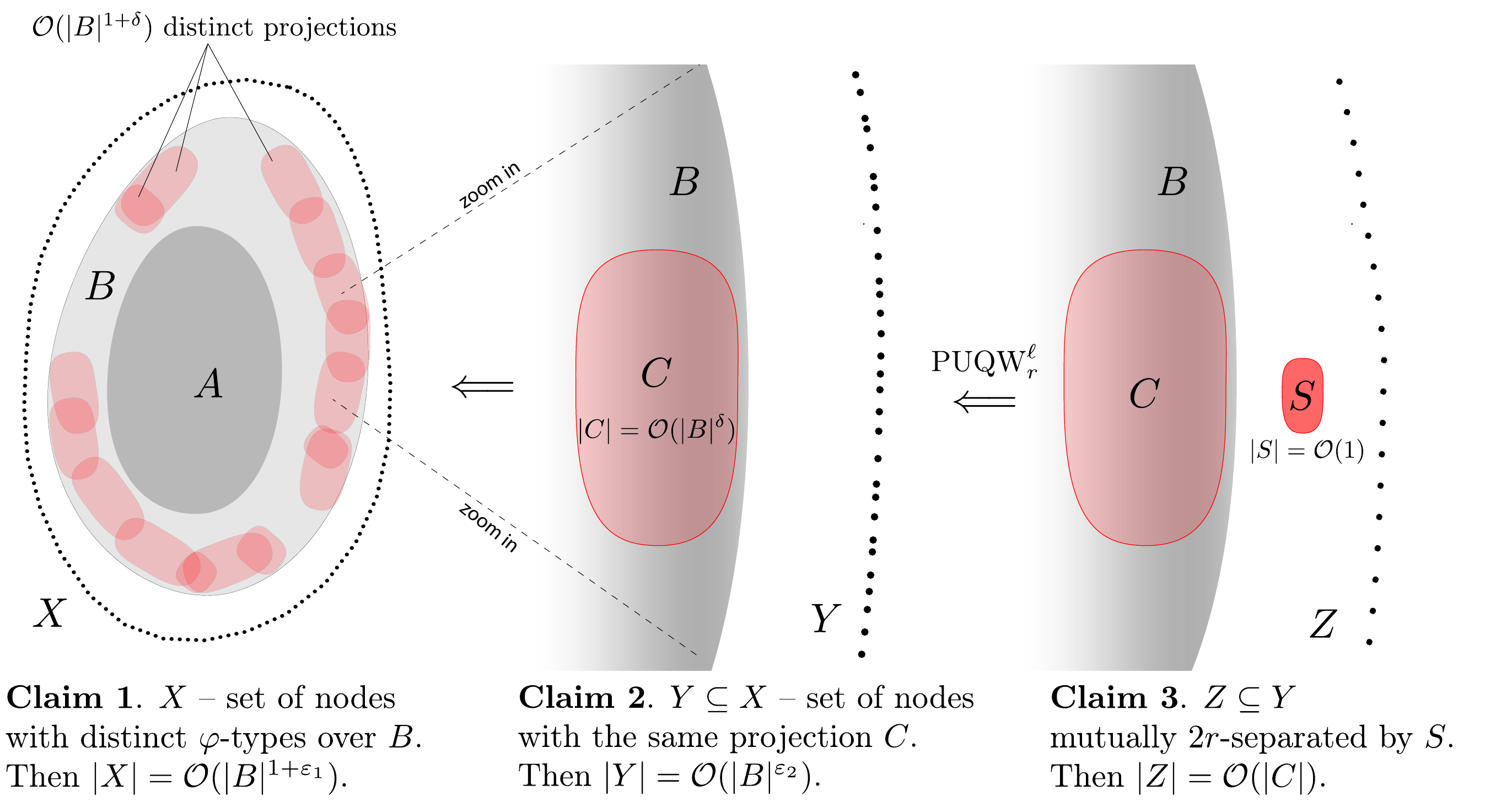}
          %code
          % \noindent\textbf{Claim 1}. $X$ -- set of nodes\\
%           with distinct $\varphi$-types over $B$.\\
%           Then $|X|={\cal O}(|B|^{1+\varepsilon_1})$.
%
%           \medskip
%           \noindent\textbf{Claim 2}. $Y\subseteq X$ -- set of nodes\\
%           with the same projection $C$.\\
%           Then $|Y|={\cal O}(|B|^{\varepsilon_2})$.
%
%           \medskip
%           \noindent\textbf{Claim 3}. $Z\subseteq Y$\\
%           mutually $2r$-separated by $S$.\\
%           Then $|Z|={\cal O}(|C|)$.
  			\caption{The proof of~\cref{thm:vc-density} in case $\ell=1$. 
  The logical implications flow from right to left,
  but our description below proceeds in the other direction.
  			}
  	   	\label{fig:sketch}
  	   \end{figure}
	
	\setcounter{claim}{0}
	
The first step is to reduce the statement to the following claim.

\begin{claim}\label{claim2}
If $X$ is a set of tuples with pairwise different $\phi$-types over $B$, then $|X|=\Oof(|B|^{\ell+\epsilon_1})$.
\end{claim}	

\cref{claim2} implies that $|S^\phi(G/B)|=\Oof(|B|^{\ell+\epsilon_1})$, 
which is bounded by $\Oof(|A|^{(\ell+\epsilon_1)(1+\delta)})$ since $|B|=\Oof(|A|^{1+\delta})$. As $(\ell+\epsilon_1)(1+\delta)\le \ell+\epsilon$, this shows that $|S^\phi(G/B)|=\Oof(|A|^{\ell+\epsilon})$.
Then \cref{lem:types-over-B} implies that also $|S^\phi(G/A)|=\Oof(|A|^{\ell+\epsilon})$, and we are done. Therefore, it remains to prove~\cref{claim2}.

\medskip

For a tuple $\bar w=w_1\ldots w_\ell\in V(G)^\ell$, define its \emph{projection}
to be the set $C_1\cup\ldots\cup C_\ell\subset B$ where  
$C_i=M^G_r(w_i, B)$. Note that there are at most 
$\Oof(|B|^{\ell(1+\delta)})$ different projections of tuples in total, and each projection has size $\Oof(|B|^\delta)$.
To prove~\cref{claim2}, we consider the special case when all the tuples have the same projection, say $C\subset B$, and  obtain a stronger conclusion,
for $\epsilon_2\coloneqq \epsilon_1-\delta\ell>0$.

\begin{claim}\label{claim3}
If $Y$ is a set of tuples with pairwise different $\phi$-types over $B$, and each $u\in V$ has the same projection $C\subset B$, then $|Y|=\Oof(|B|^{\epsilon_2})$.
\end{claim}

Since there are at most $\Oof(|B|^{\ell(1+\delta)})$ different projections in total and $\ell(1+\delta)+\epsilon_2=\ell+\epsilon_1$, \cref{claim2} can be proved
by summing the bound given by \cref{claim3} through all different projections~$C$.
It therefore remains to prove~\cref{claim3}.

\medskip

We apply~\cref{thm:uqw-tuples} to the set of $\ell$-tuples $Y$, for $m$ being the largest integer such that $|Y|\ge N^{\ell}_{2r}(m)$.
  As a conclusion, we obtain a set $Z\subseteq Y$ of $m$ tuples that is mutually $2r$-separated by $S$ in $G$, for some set of vertices $S\subseteq V(G)$ of size $s\coloneqq s^{\ell}_{2r}$.
  Let $d$ be the degree of the polynomial $N^\ell_{2r}(\cdot)$ obtained from~\cref{thm:uqw-tuples}.
  Note that $s=\Oof(1)$ and $|Y|=\Oof(m^d)$.
    
  \begin{claim}\label{claim4}
It holds that $|Z|=\Oof(|C|)$.
  \end{claim}
  
  We first show how~\cref{claim4} implies~\cref{claim3}.
  Since $m=|Z|=\Oof(|C|)$,
  and $|C|=\Oof(|B|^\delta)$,
  it follows that $|Y|=\Oof(m^d)=\Oof(|B|^{d\delta})$. As $\delta(d+\ell)>\epsilon_1$, this implies that $d\delta<\epsilon_2$, yielding~\cref{claim3}.
  We now prove~\cref{claim4}.

\medskip
  Let $Z_0\subset Z$ be the set of 
  those tuples in $Z$ which are $r$-separated by $S$ from $B$ in $G$,
  and let $Z_1=Z-Z_0$ be the remaining  tuples.
  Since tuples from $Z_0$ have pairwise different $\phi$-types over $B$, and each of them is $r$-separated by $S$ from $B$ in $G$, by~\cref{cor:bound} we infer that $|Z_0|=\Oof(1)$.  
 On the other hand, by the definition of $Z_1$, with each tuple $\bar u\in Z_1$ we may associate a vertex $b(\bar u)\in C$ which is not $r$-separated from $\bar u$ by $S$ in $G$.
 Since the set $U$ is mutually $2r$-separated by $S$ in $G$, it follows that for any two different tuples $\bar u,\bar v\in Z_1$ we have $b(\bar u)\neq b(\bar v)$.
 Hence $b(\cdot)$ is an injection from $Z_1$ to $C$, which proves that $|Z_1|\leq |C|$.
 To conclude, we have $|Z|=|Z_0|+|Z_1|=\Oof(1)+\Oof(|C|)=\Oof(|C|)$. This finishes the proof of~\cref{claim4} and ends the proof of~\cref{thm:vc-density}.

\subsection{Lower bounds for non-sparse classes}
We now move to the proof of \Cref{thm:vc-density-lower-bound},
whose statement we repeat for convenience.

 \setcounter{aux}{\thetheorem}
\setcounter{theorem}{\thevclower}
 \begin{theorem}
   Let $\CCC$ be a class of graphs which
   is closed under taking subgraphs.
   \begin{enumerate}[(1)]
   \item If $\CCC$ is not nowhere dense, then there is a formula
   $\phi(x,y)$ such that for every $n\in \N$ there are $G\in\CCC$ and $A\subseteq V(G)$
   with $|A|=n$ and $|S^\phi(G/A)|=2^{|A|}$.
   \item If $\CCC$ has unbounded expansion, then there is a formula
   $\phi(x,y)$ such that for every $c\in \mathbb{R}$ there exist $G\in\CCC$ and a nonempty $A\subseteq V(G)$ with $|S^\phi(G/A)|>c|A|$.
   \end{enumerate}
 \end{theorem}
 \setcounter{theorem}{\theaux}
%\vclower*

\begin{proof}
The first part follows easily from the following lemma.
Let $\mathcal{G}_r$ be the class of $(r-1)$-subdivisions of all 
simple graphs, that is, the class comprising
all the graphs that can be obtained from any simple graph by replacing every edge by a path of
length $r$.

\begin{lemma}[\cite{nevsetvril2011nowhere}]\label{lem:lower-nd}
For every somewhere dense graph class $\CCC$ that is closed 
under taking subgraphs, there
exists a positive integer $r$ such that $\mathcal{G}_{r}\subseteq \CCC$.
\end{lemma}

To prove the first statement of \Cref{thm:vc-density-lower-bound}, 
for $n\in \N$, let $P(n)$ denote the graph with $n+2^n$ 
vertices $V(P(n))\coloneqq \{v_1,\ldots, v_n\}\cup \{w_M \colon M\subseteq \{1,\ldots, n\}\}$ and edges $E(P(n))\coloneqq \{v_iw_M \colon 1\leq i\leq n,\, M\subseteq \{1,\ldots, n\},\, i\in M\}$. 
If $\CCC$ is somewhere dense and closed under taking subgraphs, 
according to \Cref{lem:lower-nd}, there exists an integer $r$ 
such that $\mathcal{G}_{r}\subseteq \CCC$. In particular, for every $n\in \N$ the $(r-1)$-subdivision $P^{r}(n)$ of the graph $P(n)$ is contained in $\CCC$.
Now consider 
the formula $\phi(x,y)$ stating that $x$ and~$y$ are at distance at most $r$. Then for every $n\in \N$ we have 
$S^\phi(P^{r}(n)/A)=\Pow(A)$, where $A\subseteq V(P^{r}(n))$ denotes the set $\{v_1,\ldots, v_n\}$. This implies the first part
of the theorem. 

\medskip
We now move to the  second part of~\cref{thm:vc-density-lower-bound}.
A graph $H$ is a \emph{topological depth-$r$ minor} of~$G$ if
there is a mapping $\phi$ that maps vertices of~$H$ to 
vertices of $G$ such that $\phi(u)\neq \phi(v)$ for 
$u\neq v$, and edges of $H$ to paths in 
$G$ such that if $uv\in E(H)$, then $\phi(uv)$
is a path of length at most $2r+1$ between $u$ and $v$ in 
$G$ and furthermore, if $uv, xy\in E(H)$, then 
$\phi(uv)$ and $\phi(xy)$ are internally vertex
disjoint. We write $H\minor_r^{\mathrm{t}} G$. 
Note that the above definition makes sense for 
half-integers, i.e., numbers $r$ for which $2r$ is an integer.

It is well-known that classes of bounded expansion can be alternatively characterized by the sparsity of shallow topological minors.

\begin{lemma}[Corollary 4.1 of \cite{sparsity}]\label{lem:top-bnd-exp}
A class $\CCC$ of graphs has bounded expansion if and only 
if for every $r\in \N$ there exists a constant $c_r$ such that $|E(H)|/|V(H)|\leq c_r$ for all graphs $H$ such that $H\minor_r^{\mathrm{t}} G $ for some $G\in \CCC$.
\end{lemma}

For $r\in \N$ and a graph $G$, by $\nu_r(G)$ we denoted the
\emph{normed $r$-neighborhood complexity} of $G$, as defined
by Reidl et al.~\cite{reidl2016characterising}: 
\begin{equation*}
\nu_r(G)\coloneqq\max_{H\subseteq G,\,\emptyset\neq A\subseteq V(G)}\frac{|\{N_r^H[v]\cap A\, \colon\, v\in V(H)\}|}{|A|}.
\end{equation*}
We will need the following result relating edge density in shallow topological minors and normed neighborhood complexity.

\begin{lemma}[Theorem 4 of \cite{reidl2016characterising}]\label{lem:lower-be}
Let $G$ be a graph, $r$ be a half-integer, 
and let $H\minor_r^{\mathrm{t}}G$. 
Then 
$$\frac{|E(H)|}{|V(H)|}\leq (2r + 1)\cdot \max \left\{\nu_1(G)^4\cdot \log^2\nu_1(G),\nu_2(G),\ldots, \nu_{\left\lceil r+\frac{1}{2}\right\rceil}(G)\right\}.$$
\end{lemma}

For the second part of~\cref{thm:vc-density-lower-bound}, we use the contrapositive of \Cref{lem:top-bnd-exp}. Since $\CCC$ has unbounded expansion, for some $r\in \N$ 
we have that the value $|E(H)|/|V(H)|$ is unbounded among depth-$r$ topological minors $H$ of graphs from $\CCC$.
By applying \Cref{lem:lower-be}, we find that for some $q\leq r$ the value
$\nu_{q}(G)$ is unbounded when $G$ ranges over all graphs from $\CCC$. 
Since $\CCC$ is closed under taking subgraph, we infer that also the ratio 
$\frac{|\{N_q^G[v]\cap A \colon v\in V(G)\}|}{|A|}$ is unbounded when~$G$ ranges over graphs from $\CCC$ and $A$ ranges over nonempty subsets of $V(G)$.
This is equivalent to the sought assertion for the formula $\phi(x,y)$ expressing that $x$ and~$y$ are at distance at most $q$. 

\mbox{}
\end{proof}

\section{Packing and traversal numbers for nowhere dense graphs}\label{sec:ep}
In this section, we give an application 
of~\cref{thm:vc-density}, proving a 
duality result for nowhere dense graph classes.

A \emph{set system} is a family  $\cal F$ of subsets of a set $X$.
Its  \emph{packing} is a subfamily of $\cal F$ of pairwise disjoint subsets, and its \emph{traversal} (or \emph{hitting set}) is a subset of $X$ which intersects every member of $\cal F$.
The \emph{packing number} of~$\cal F$, denoted $\nu(\cal F)$, is the largest cardinality of a packing in $\cal F$,
and the \emph{transversality} of $\cal F$, denoted
$\tau(\cal F)$, is the smallest cardinality of a traversal of $\cal F$.
Note that if $\cal G$ is a finite set system, then
$\nu({\cal G})\le \tau(\cal G)$. 
The set system $\cal F$ has the \emph{Erd{\H o}s-P\'{o}sa property} if there is a function $f\from\N\to\N$ such that every finite subfamily $\cal G$ of $\cal F$
satisfies $\tau({\cal G})\le f(\nu(\cal G))$. 

We prove that set systems defined by first order formulas in nowhere dense graph classes have the Erd{\H o}s-P\'{o}sa property, in the following sense.

 \setcounter{aux}{\thetheorem}
 \setcounter{theorem}{\theep}
\begin{theorem}
	Fix a nowhere dense class of graphs $\CCC$ and a 
	formula $\phi(x,y)$ with two free variables $x,y$.
	Then there is a function $f\from \N\to\N$ with the following property.
	Let $G\in \CCC$ be a graph and let $\cal G$
	be a family of subsets of $V(G)$ consisting of sets of the form $\setof{v\in V(G)}{\phi(u, v)}$, where~$u$ is some vertex of $V(G)$.
Then~$\tau({\cal G})\le f(\nu(\cal G))$.
\end{theorem}
 \setcounter{theorem}{\theaux}

%\erdosposa*
% \begin{theorem}\label{thm:erdos-posa}
% 	Fix a nowhere dense class of graphs $\CCC$ and a
% 	formula $\phi(x,y)$ with free variables $x,y$.
% 	 % where $x$ is a single variable and $\bar y$ is a tuple of variables.
% 	There is a function $f\from \N\to\N$ with the following property.
% 	Let $G\in \CCC$ be a graph and let $\cal G$
% 	be a family of subsets of $V(G)$ consisting of sets of the form $\setof{b\in V(G)}{\phi(a, b)}$, where~$a\in V(G)$.
% Then~$\tau({\cal G})\le f(\nu(\cal G))$.
% \end{theorem}

We will apply the following result of Matou{\v s}ek~\cite{Matousek:2004:BVI:1005787.1005789},
which relies on the proof of Alon and Kleitman~\cite{ALON1992103} of the conjecture of Hardwiger and Debrunner. 
In the result of Matou{\v s}ek, the set system $\cal F$ is infinite. For $m\in \N$, by $\pi_{\cal F}^*(m)$ we denote the \emph{dual shatter function} of $\cal F$, which is defined as the maximal number 
of occupied cells in the Venn diagram of $m$ sets in $\cal F$.

\begin{theorem}[Matou{\v s}ek, \cite{Matousek:2004:BVI:1005787.1005789}]\label{thm:pq}
	Let $\cal F$ be a set system with $\pi^*_{\cal F}(m)=o(m^k)$,
	for some integer $k$, and let $p\ge k$.
	Then there is a constant $T$ such that the following holds for every finite family $\cal G\subset \cal F$: 
	if $\cal G$ has the $(p,k)$-property, meaning that 
	among every $p$ sets in $\cal G$ some $k$ have a non-empty intersection, then $\tau ({\cal G})\le T$.
\end{theorem}
\begin{proof}[of \cref{thm:erdos-posa}]
For a graph $G$, define the set system ${\cal F}_G$ on the ground set $V(G)$ as
$${\cal F}_G = \setof{\setof{v\in V(G)}{\phi(u, v)}}{u\in V(G)}.$$
Let then $\cal F$ be the disjoint union of set systems ${\cal F}_G$ for $G\in \CCC$. That is, 
the ground set of $\cal F$ is the disjoint union of the vertex sets $V(G)$ for $G\in \CCC$, and for each $G\in \CCC$ we add to ${\cal F}$
a copy of ${\cal F}_G$ over the copy of relevant $V(G)$.
Then the following claim follows directly from~\cref{thm:vc-density}.

\begin{claim}
The dual shatter function of $\cal F$ satisfies $\pi^*_{\cal F}(m)=\Oof(m^{1+\epsilon})$,
for every fixed $\epsilon>0$. In particular, $\pi^*_{\cal F}(m)=o(m^{2})$.
\end{claim}

Consider the function $f\from \N \to \N$ defined so that $f(\nu)$ is the value $T$ obtained from~\cref{thm:pq} applied to $\cal F$, $k=2$, and $p=\nu+1$.
Suppose now that $G\in \CCC$ is a graph and $\GGG\subseteq \FFF_G$
is a family of subsets of $V(G)$ consisting of sets of the form $\{v\in V(G)\,\colon\,\phi(u,v)\}$, where $u$ is some vertex of $G$.
We identify $\GGG$ with a subfamily of $\FFF$ in the natural way, following the embedding of $\FFF_G$ into $\FFF$ used in the construction of the latter.
Let $\nu$ be the packing number of $\GGG$.
In particular, for every $\nu+1$ subsets of $\GGG$
there is a vertex $v\in V(G)$
which is contained in two elements of~$\GGG$.
Hence, $\GGG$ is a $(p,2)$-family for $p=\nu+1$.
By~\cref{thm:pq}, $\tau(\GGG)\le T=f(\nu)=f(\nu(\GGG))$, as required.
\end{proof}

\section{Bounds for stability}\label{sec:stable}
As mentioned, Adler and Adler~\cite{adler2014interpreting}, 
proved that every nowhere dense class of graphs is stable. In this section,
we prove its effective variant,~\cref{thm:new-stable}, which 
we repeat for convenience.
%\newstable*
 \setcounter{aux}{\thetheorem}
 \setcounter{theorem}{\thestable}
 \begin{theorem}
 There are computable functions $f\colon \N^3\to\N$ and $g\colon\N\to\N$ with the following property.
 Suppose $\phi(\bar x,\bar y)$ is a formula of quantifier rank at most $q$ and with $d$ free variables.
 Suppose further that $G$ is a graph excluding $K_t$ as a depth-$g(q)$ minor. Then the ladder index of $\phi(\bar x,\bar y)$ in $G$ is at most $f(q,d,t)$.
 \end{theorem}
 \setcounter{theorem}{\theaux}

Recall that a class $\cal C$ is stable if and only if for every first order formula $\varphi(\bar x,\bar y)$, 
its ladder index over graphs from $\cal C$ is bounded by a constant depending only on $\cal C$ and $\varphi$;
see \cref{sec:intro} to recall the background on stability.
Thus the result of Adler and Adler is implied by \cref{thm:new-stable},
and is weaker in the following sense: \cref{thm:new-stable} asserts in addition that there is a computable bounds on the ladder index
of any formula that depends only on the size of an excluded clique minor at depth bounded in terms of formula's quantifier rank and number of free variables. 
We now prove~\cref{thm:new-stable}.

\begin{proof}[of~\cref{thm:new-stable}]
Fix a formula $\phi(\bar x,\bar y)$ of quantifier rank $q$ and
a partitioning of its 
free variables into  $\bar y$ and $\bar z$.
Let $d=|\bar x|+|\bar y|$ be the total number of free variables of $\phi$.
Let $r\in \N$ be the number given by~\cref{cor:bound},
which depends on $\phi$ only.
Let $\cal C$ be the class of all graphs 
such that  $K_t\not\minor_{18r} G$.
Then, by~\cref{thm:uqw-tuples}, 
$\cal C$ satisfies $\uqw^d_{r}(N^d_{r},s^d_r)$,
for some  polynomial  $N^d_r\from\N\to\N$ and number $s=s^d_r\in \N$ computable from $d,t,r$.
Let $T$ be the number given by~\cref{cor:bound} for $\phi$ and $s$.
 Finally, let 
$\ell=N^d_r(2T+1)$.
We show that 
every $\phi$-ladder in a graph $G\in\cal C$ has length smaller than $\ell$.

For the sake of contradiction, assume that there is a graph $G\in\cal C$
and tuples $\bar u_1,\ldots,\bar u_\ell\in V(G)^{|\bar x|}$ and $ \bar v_1,\ldots, \bar v_\ell\in V(G)^{|\bar y|}$
which form a $\phi$-ladder in $G$, i.e., 
$\phi(\bar u_i,\bar v_j)$ holds in~$G$ if and only if $i\le j$.
	Let $A=\setof{ \bar u_i \bar v_i}{i=1,\ldots,\ell}\subset V(G)^d$. Note that $|A|=\ell\ge N^d_r(2T+1)$, since tuples $\bar u_i$ have to be pairwise different.
  
Applying property  $\uqw^d_r(N^d_r,s^d_r)$ to the set $A$, radius $r$, and target size $m=2T+1$
		 yields a set $S\subset V(G)$ with $|S|\le s$
	and a set $B\subset A$ with $|B|\geq 2T+1$ 
  of tuples which are  mutually $r$-separated by $S$  in $G$.
  Let $J\subseteq \set{1,\ldots,\ell}$
  be the set of indices corresponding to $B$,
  i.e., $J=\set{j\colon\bar u_j\bar v_j\in B}$.
  
  Since $|J|=2T+1$, we may partition $J$ into $J_1\uplus J_2$ with $|J_1|=T+1$ so that the following condition holds:
  for each $i,k\in J_1$ satisfying $i<k$, there exists $j\in J_2$ with $i<j<k$. Indeed, it suffices to order the indices of $J$ and put every second index to $J_1$, and every other to $J_2$.
  Let~$X$ be the set of vertices appearing in the tuples $\bar u_i$ with $i\in J_1$, and let $Y$ be the set of vertices appearing in the tuples $\bar v_j$ with $j\in J_2$.
  Since the tuples of $B$ are mutually $r$-separated by $S$ in~$G$, it follows that $X$ and $Y$ are $r$-separated by $S$.
  As $|J_1|=T+1$, by \cref{cor:bound} we infer that there are distinct indices $i,k\in J_1$, say $i<k$, such that $\tp^\phi(\bar u_i/Y)=
    \tp^\phi(\bar u_{k}/Y)$. This implies that for each $j\in J_2$, we have $G,\bar u_i,\bar v_j\models \phi(\bar x,\bar y)$ if and only if $G,\bar u_{k},\bar v_j\models \phi(\bar x,\bar y)$.
    However, there is an index $j\in J_2$ such that $i<j<k$, and for this index we should have $G,\bar u_i,\bar v_j\models \phi(\bar x,\bar y)$ and $G,\bar u_{k},\bar v_j\not\models \phi(\bar x,\bar y)$
    by the definition of a ladder. This contradiction concludes the proof.
\end{proof}

We remark that~\cref{thm:new-stable} also holds for classes of edge- and vertex-colored graphs, with the same proof, but using 
a version of the results in~\cref{sec:gaifman} lifted to edge- and vertex-colored graphs.

\bibliographystyle{abbrv}
\bibliography{ref}

\begin{thebibliography}{10}

\bibitem{adler2014interpreting}
H.~Adler and I.~Adler.
\newblock Interpreting nowhere dense graph classes as a classical notion of
  model theory.
\newblock {\em European Journal of Combinatorics}, 36:322--330, 2014.

\bibitem{ALON1992103}
N.~Alon and D.~J. Kleitman.
\newblock Piercing convex sets and the hadwiger-debrunner (p, q)-problem.
\newblock {\em Advances in Mathematics}, 96(1):103 -- 112, 1992.

\bibitem{alon2003turan}
N.~Alon, M.~Krivelevich, and B.~Sudakov.
\newblock Tur{\'a}n numbers of bipartite graphs and related {R}amsey-type
  questions.
\newblock {\em Combinatorics, Probability and Computing}, 12(5+ 6):477--494,
  2003.

\bibitem{aschenbrenner2016vapnik}
M.~Aschenbrenner, A.~Dolich, D.~Haskell, D.~Macpherson, and S.~Starchenko.
\newblock Vapnik-{C}hervonenkis density in some theories without the
  independence property, {I}.
\newblock {\em Transactions of the American Mathematical Society},
  368(8):5889--5949, 2016.

\bibitem{atserias2006preservation}
A.~Atserias, A.~Dawar, and P.~G. Kolaitis.
\newblock On preservation under homomorphisms and unions of conjunctive
  queries.
\newblock {\em Journal of the ACM (JACM)}, 53(2):208--237, 2006.

\bibitem{BousquetT15}
N.~Bousquet and S.~Thomass{\'{e}}.
\newblock {VC}-dimension and {E}rd{\H{o}}s-{P}{\'{o}}sa property.
\newblock {\em Discrete Mathematics}, 338(12):2302--2317, 2015.

\bibitem{Bronnimann1995}
H.~Br{\"{o}}nnimann and M.~T. Goodrich.
\newblock Almost optimal set covers in finite {VC}-dimension.
\newblock {\em Discrete {\&} Computational Geometry}, 14(4):463--479, 1995.

\bibitem{chervonenkis1971theory}
A.~Chervonenkis and V.~Vapnik.
\newblock Theory of uniform convergence of frequencies of events to their
  probabilities and problems of search for an optimal solution from empirical
  data.
\newblock {\em Automation and Remote Control}, 32:207--217, 1971.

\bibitem{dawar2010homomorphism}
A.~Dawar.
\newblock Homomorphism preservation on quasi-wide classes.
\newblock {\em Journal of Computer and System Sciences}, 76(5):324--332, 2010.

\bibitem{diestel2012graph}
R.~Diestel.
\newblock {\em Graph Theory, 4th Edition}, volume 173 of {\em Graduate {T}exts
  in {M}athematics}.
\newblock Springer, 2012.

\bibitem{drange2016kernelization}
P.~G. Drange, M.~S. Dregi, F.~V. Fomin, S.~Kreutzer, D.~Lokshtanov,
  M.~Pilipczuk, M.~Pilipczuk, F.~Reidl, F.~{S{\'{a}}nchez Villaamil},
  S.~Saurabh, S.~Siebertz, and S.~Sikdar.
\newblock Kernelization and sparseness: the case of {D}ominating {S}et.
\newblock In {\em {STACS 2016}}, volume~47 of {\em LIPIcs}, pages 31:1--31:14.
  Schloss Dagstuhl---Leibniz-Zentrum f\"ur Informatik, 2016.
\newblock See \url{https://arxiv.org/abs/1411.4575} for full proofs.

\bibitem{dvovrak2013testing}
Z.~Dvo{\v{r}}{\'a}k, D.~Kr{\'a}l, and R.~Thomas.
\newblock Testing first-order properties for subclasses of sparse graphs.
\newblock {\em Journal of the ACM (JACM)}, 60(5):36, 2013.

\bibitem{eickmeyer2016neighborhood}
K.~Eickmeyer, A.~C. Giannopoulou, S.~Kreutzer, O.~Kwon, M.~Pilipczuk,
  R.~Rabinovich, and S.~Siebertz.
\newblock Neighborhood complexity and kernelization for nowhere dense classes
  of graphs.
\newblock In {\em {ICALP 2017}}, volume~80 of {\em LIPIcs}, pages 63:1--63:14.
  Schloss Dagstuhl---Leibniz-Zentrum f\"ur Informatik, 2017.
\newblock See \url{https://arxiv.org/abs/1612.08197} for full proofs.

\bibitem{furedi1991traces}
Z.~F{\"u}redi and J.~Pach.
\newblock Traces of finite sets: extremal problems and geometric applications.
\newblock {\em Extremal problems for finite sets}, 3:255--282, 1991.

\bibitem{gaifman1982local}
H.~Gaifman.
\newblock On local and non-local properties.
\newblock {\em Studies in Logic and the Foundations of Mathematics},
  107:105--135, 1982.

\bibitem{grokre11}
M.~Grohe and S.~Kreutzer.
\newblock Methods for algorithmic meta theorems.
\newblock In M.~Grohe and J.~Makowsky, editors, {\em Model Theoretic Methods in
  Finite Combinatorics}, volume 558 of {\em Contemporary Mathematics}, pages
  181--206. American Mathematical Society, 2011.

\bibitem{GroheKRSS15}
M.~Grohe, S.~Kreutzer, R.~Rabinovich, S.~Siebertz, and K.~Stavropoulos.
\newblock Colouring and covering nowhere dense graphs.
\newblock In {\em {WG 2015}}, volume 9224 of {\em Lecture Notes in Computer
  Science}, pages 325--338. Springer, 2015.

\bibitem{grohe2014deciding}
M.~Grohe, S.~Kreutzer, and S.~Siebertz.
\newblock Deciding first-order properties of nowhere dense graphs.
\newblock In {\em STOC 2014}, pages 89--98. ACM, 2014.

\bibitem{ivanov}
A.~A. Ivanov.
\newblock The structure of superflat graphs.
\newblock {\em Fundamenta Mathematicae}, 143:107--117, 1993.

\bibitem{kreidler1998monadic}
M.~Kreidler and D.~Seese.
\newblock Monadic {NP} and graph minors.
\newblock In {\em {CSL 1998}}, volume 1584 of {\em Lecture Notes in Computer
  Science}, pages 126--141. Springer, 1998.

\bibitem{siebertz2016polynomial}
S.~Kreutzer, R.~Rabinovich, and S.~Siebertz.
\newblock Polynomial kernels and wideness properties of nowhere dense graph
  classes.
\newblock In {\em {SODA 2017}}, pages 1533--1545. {SIAM}, 2017.

\bibitem{laskowski1992vapnik}
M.~C. Laskowski.
\newblock Vapnik-{C}hervonenkis classes of definable sets.
\newblock {\em Journal of the London Mathematical Society}, 2(2):377--384,
  1992.

\bibitem{malliaris2014regularity}
M.~Malliaris and S.~Shelah.
\newblock Regularity lemmas for stable graphs.
\newblock {\em Transactions of the American Mathematical Society},
  366(3):1551--1585, 2014.

\bibitem{matouvsek1998geometric}
J.~Matou{\v{s}}ek.
\newblock Geometric set systems.
\newblock In {\em European Congress of Mathematics}, volume~2, page~23.
  Birkh{\"a}user, Basel, 1998.

\bibitem{Matousek:2004:BVI:1005787.1005789}
J.~Matou{\v{s}}ek.
\newblock Bounded {VC}-dimension implies a fractional {H}elly theorem.
\newblock {\em Discrete {\&} Computational Geometry}, 31(2):251--255, 2004.

\bibitem{nevsetvril2008grad}
J.~Ne{\v{s}}et{\v{r}}il and P.~Ossona~de Mendez.
\newblock Grad and classes with bounded expansion {I}. {D}ecompositions.
\newblock {\em European Journal of Combinatorics}, 29(3):760--776, 2008.

\bibitem{nevsetvril2010first}
J.~Ne{\v{s}}et{\v{r}}il and P.~Ossona~de Mendez.
\newblock First order properties on nowhere dense structures.
\newblock {\em The Journal of Symbolic Logic}, 75(03):868--887, 2010.

\bibitem{nevsetvril2011nowhere}
J.~Ne{\v{s}}et{\v{r}}il and P.~Ossona~de Mendez.
\newblock On nowhere dense graphs.
\newblock {\em European Journal of Combinatorics}, 32(4):600--617, 2011.

\bibitem{sparsity}
J.~Ne\v{s}et\v{r}il and P.~{Ossona de Mendez}.
\newblock {\em Sparsity --- {G}raphs, {S}tructures, and {A}lgorithms},
  volume~28 of {\em Algorithms and combinatorics}.
\newblock Springer, 2012.

\bibitem{pillay}
A.~Pillay.
\newblock {\em Introduction to Stability Theory}.
\newblock Dover Books on Mathematics. Dover Publications, 2008.

\bibitem{podewski1978stable}
K.-P. Podewski and M.~Ziegler.
\newblock Stable graphs.
\newblock {\em Fundamenta Mathematicae}, 100(2):101--107, 1978.

\bibitem{reidl2016characterising}
F.~Reidl, F.~{S{\'{a}}nchez Villaamil}, and K.~Stavropoulos.
\newblock Characterising bounded expansion by neighbourhood complexity.
\newblock {\em CoRR}, abs/1603.09532, 2016.

\bibitem{sauer1972density}
N.~Sauer.
\newblock On the density of families of sets.
\newblock {\em Journal of Combinatorial Theory, Series~A}, 13(1):145--147,
  1972.

\bibitem{shelah1972combinatorial}
S.~Shelah.
\newblock A combinatorial problem; stability and order for models and theories
  in infinitary languages.
\newblock {\em Pacific Journal of Mathematics}, 41(1):247--261, 1972.

\bibitem{shelah1990classification}
S.~Shelah.
\newblock {\em Classification theory: and the number of non-isomorphic models},
  volume~92 of {\em Studies in Logic and the Foundations of Mathematics}.
\newblock Elsevier, 1990.

\bibitem{tent2012course}
K.~Tent and M.~Ziegler.
\newblock {\em A Course in Model Theory}.
\newblock Lecture Notes in Logic. Cambridge University Press, 2012.

\bibitem{zhu2009coloring}
X.~Zhu.
\newblock Colouring graphs with bounded generalized colouring number.
\newblock {\em Discrete Mathematics}, 309(18):5562--5568, 2009.

\end{thebibliography}

\end{document}